%% file: main.tex
\documentclass[a4paper,12pt]{article}
\usepackage[style=mla]{biblatex}
\usepackage{graphicx} 
\usepackage[margin=1in]{geometry}
\usepackage{setspace}
\usepackage{hyperref}
\usepackage{booktabs}
\usepackage{mathtools}
\usepackage{amsmath}
\usepackage{amsthm}
\usepackage{amssymb}
\usepackage{amsfonts}
\usepackage{amssymb}
\usepackage{bbm}
\usepackage{titling}
\usepackage{lettrine}
\usepackage[font=small,labelfont=bf]{caption}
\usepackage{sgame}
\usepackage{mdframed}
\newcommand{\citelink}[2]{\hyperlink{cite.\therefsection @#1}{#2}}
\newcommand{\qcite}[1]{\citelink{#1}{\citeauthor{#1} (\citeyear{#1})}}
\newcommand{\qcitenp}[1]{\citelink{#1}{\citeauthor{#1} \citeyear{#1}}}
\newcommand{\ubar}[1]{\text{\b{$#1$}}}

\hypersetup{
    colorlinks=true,
    linkcolor=blue,
    filecolor=magenta,      
    urlcolor=cyan
}

\usepackage{enumitem} 
\usepackage{calc}

\usepackage{tikz}
\usetikzlibrary{arrows}
\usetikzlibrary{positioning}

\newtheorem{theorem}{Theorem}
\newtheorem{proposition}{Proposition}
\theoremstyle{definition}

\newtheorem{corollary}{Corollary}
\newtheorem{assumption}{Assumption}

\usepackage{float}



\title{A Dynamic Matching Framework for Faster Child Adoptions}
\author{Terence Highsmith II\thanks{Kellogg School of Management at Northwestern University (email: \url{terence.highsmith@kellogg.northwestern.edu}). I thank James Schummer and George Georgiadis for invaluable guidance throughout this project as my advisors. I also thank Nemanja Antic, Alexander Jakobsen, Benjamin Friedrich, Nicola Bianchi, Daniel Barron, Guilherme Neri, Victor Matta, and Michael Powell for insightful feedback and comments.}}
\date{\vspace{-3em}}

\begin{document}
\maketitle
\begin{center}
    \begin{minipage}{0.9\textwidth}
\textit{Caseworkers in foster care systems match waiting children to adoptive homes. We use dynamic matching market design to characterize a class of mechanisms that incentivize expedient matches that homes can accept or decline. We design mechanisms satisfying fairness and limited strategy-proofness. They also avoid costly patience. Our empirically-based simulations suggest the mechanisms could increase adoptions by at least 25\% versus the status quo. A naive dynamic extension of Deferred Acceptance does not attain these benefits. Our mechanisms sidestep direct preference elicitation by predicting preferences, and they are robust to prediction error.} (JEL D47)	
\end{minipage}
\end{center}

\vspace{1em}

\input{Paper/Introduction}

\input{Paper/Model}

\input{Paper/DynamicEnvyFree}

\input{Paper/StrategyProofness}

\input{Paper/Empirics}

\input{Paper/Conclusion}

\input{Appendices/A}

\input{Appendices/B}

\printbibliography

\end{document}

%% file: Paper/Introduction.tex
\lettrine[lraise=0.3, lines=2, findent=3pt, nindent=0pt]{T}{he} growing literature on dynamic matching theory references foster care and child adoption as a prime example of a dynamic matching market. Nevertheless, there are no applicable mechanisms to match waiting children with adoptive homes. We break this theory-application barrier. Existing tools from dynamic matching do not satisfy the necessary constraints for adoption from child welfare systems, and existing tools from the field struggle to provide consistent, efficient outcomes. Our paper develops a fitting framework and axioms for when children and homes arrive to the system over time, the authority may unilaterally allow or disallow some matches, and homes can choose to accept or decline the authority's proposed match. 

In 2022, about 109 thousand children were waiting for adoption. The median child had been waiting for over 29 months (\qcitenp{afcars}). A systemic review of outcomes for children languishing in foster care aptly states, "outcomes of foster youth are troubling on all domains" (\qcitenp{gypen-2017}). Interwoven mechanisms drive longer waits for these children, but one clear flaw is a lack of systems that prioritize efficient matching. At present, the process is very labor intensive, and only some counties responsible for waiting children use matching tools that could expedite the process. The matching tools that do exist are often simple spreadsheets or predictive measures that do not consider capacity utilization, fairness, nor speed. \qcite{hansen-2006} note that in some cases, states have had oversupplies of willing adopters while simultaneously maintaining large populations of waiting children.

Using these institutional drawbacks as motivation, we develop matching mechanisms that practitioners can employ to better dynamically manage waiting children populations and achieve more and faster adoptions. In our framework, we model a dynamic matching market where children and homes arrive over time. As the matchmaker effectively has unilateral rights over placing children in homes, we model the homes' strategic decisions as accepting or declining proposed adoptions\footnote{Child welfare authorities and children are well-aligned. The authorities have enormous leeway in placement decisions when children are in foster care, and children, not having a firm idea of how to evaluate a home, most often follow the authorities' decisions. If a child prefers some homes over others, the authorities take this into consideration in their own preferences.}. Children and homes both have cardinal utility and waiting costs; these forces create dynamic tradeoffs between accepting a subpar match in the present versus waiting for the perfect match in the future. We focus on analyzing which classic properties our mechanisms can satisfy: fairness (stability), no unfilled capacity, individual rationality, and strategy-proofness.

We introduce important new properties for matching mechanisms: patience-freeness, weak fairness, (weak) non-wastefulness, and capacity efficiency. Given time discounting, a matching mechanism is patience-free if, after any history of matches $m_{t-1}$, the matching that the mechanism produces at time $t$ is weakly less preferred to earlier matches at $k < t$ for all homes. Implementing patience-freeness is impossible if the mechanism is fair (Proposition \ref{pr:noexistence}). Instead, we focus on weakly fair mechanisms that produce matchings with no child and home pair, where the home is matched, that would mutually prefer each other rather than their matches. These two properties allow the matchmaker to use delays or less preferred placements to create dynamic incentives for homes to accept early placements. Weak non-wastefulness and capacity efficiency then restrict the mechanism to matches that do not allow too many unmatched homes\footnote{Whereas, a fair matching always matches homes whenever possible.}. We design a mechanism that efficiently utilizes capacity in the class of weakly fair and patience-free mechanisms (Theorem \ref{th:ADA}).

Mechanisms creating dynamic incentives for homes to accept earlier placements appear similar to policies implemented in a select few counties, showing that they may be feasible in practice\footnote{These penalties are rare in child welfare, but they have appeared in at least one state. In \qcite{slaugh-2016}, the authors note that the Pennsylvania state matchmaker implemented a policy to wait thirty days between successive placements for the same home when the home rejected a placement.}. We also prove that patience-free mechanisms induce a weakly dominant strategy for homes to accept the first placement that the matchmaker offers (Proposition \ref{pr:dominant}). Unlike many other stability notions in dynamic matching literature, a weakly fair and patience-free mechanism always exists (Corollary \ref{cr:existence}). In the second half of our theoretical developments, we provide multiple results on the strategic incentives of our mechanisms. We find that not all are strategy-proof: homes might misrepresent their desire to accept certain placements. We devise a mechanism that attains limited strategy-proofness (Theorem \ref{th:HEDA}), although manipulable mechanisms may marginally increase placements.

We use a national dataset to simulation counterfactual outcomes under our mechanisms. Our simulations sample random counties, children, and homes from the empirical distribution to "replay" the matching market under our mechanisms, sequential Deferred Acceptance (SDA), and decentralized matching. We obtain data from \qcite{highsmith2024matching} that includes an objective measure for children's welfare, and we simulate multiple potential preferences for homes. We find that patience-free mechanisms create at least 25\% more placements versus decentralized matching, and whereas we cannot statistically distinguish SDA from decentralized matching in most preference models. Our mechanisms increase placements for teenagers by at least 16\% when homes have vertical preferences inversely correlated with age. They increase placements for children with disabilities at least 14\% when homes have weakly negative, horizontally differentiated preferences over disability status. Both groups of children typically wait longest for adoptions under the status quo.

Our theoretical developments show that although some strategic preference reports can be pre-empted, we cannot elicit exact cardinal values. Instead, our mechanisms use a simple estimator to predict preferences that any practitioner can feasibly implement. Even when the estimator is only loosely correlated with true preferences, our mechanisms still out-perform decentralized matching and SDA, maintaining the lower bounds mentioned above. Our mechanisms likewise significantly decrease the costs that counties bear while caring for waiting children.

Lastly, we provide insights on practically implementing our system within child welfare systems. We discuss strategies and logistical challenges for switching from decentralized to centralized matching, the advantages of predicting preferences and methods to optimize prediction, and potential concerns for strategic manipulation when using inaccurate predictors.

Theoretically, our model framework borrows prominently from \qcite{baccara-2020} and our approach to the solution concept from \qcite{doval-2022}. The former paper analyzes the optimal market thickness in two-sided dynamic matching problems. We add generality in agents' preferences and agents' arrival to the market to capture relevant aspects of child adoption from foster care. In practice, county authorities have little to no desire to wait for thicker markets. Instead, they would prefer to expedite children's exits from the foster care system. Timely participation is exactly the focus of the latter paper. It departs from optimal dynamic matching to spotlight dynamic stability, a concept that ensures that both sides of the market follow the matchmaker's recommendation at the right times. However, dynamic stability is stringent, and, as Doval points out, whether or not an algorithm to compute dynamically stable matchings exists is an open question. Our setting allows us to use a more relaxed fairness concept because one side of the market---the county authority---is a guaranteed participant and can block homes from achieving undesired placements.

Our framework applies when any authority wants to motivate faster matches. Homeless shelters are one example. Homeless clients in crisis can contact county services to request shelter, and the county authority chooses a shelter to match the request. Evaluation reports from counties suggest that the third most common reason for request denials is shelter eligibility requirements, and the same reports cite this as a challenge (\qcitenp{barthel2020missed}). A shelter can require clients to have a specific level of income, age, gender, or a number of other characteristics. A client with a request today that is acceptable might be refused in hopes that a better client will be referred tomorrow. In a context where expedient matching is critical, this is undesirable. Another case is deceased-donor organ allocation. Much like adoption, many patients in need experience long wait times and shortages. Direct evidence shows that 13\% of deceased-donor kidneys are rejected over 100 times before transplantation, and these delays lead to discarded organs (\qcitenp{king2021declined}). Our mechanisms can curb dynamic incentives for waiting for more preferred clients or kidneys with minimal, fair wastefulness. 

\input{Components/InstitutionalDetails}

\input{Components/LiteratureContribution}

\input{Components/Outline}

%% file: Components/InstitutionalDetails.tex
\subsubsection*{Institutional Details}

In the United States, children that are victims of abuse or neglect enter into the local child welfare---alternatively, foster care---system. Typically, the local county manages these systems. Once a child is in the system, they enter the management of a caseworker. A worker's first priority is to reunify the child with the home of origin after rectifying the situation that warranted the child entering foster care. Failing this, many counties attempt to place the child with kin of the family of origin. Finally, if the county cannot reunify or place with kin, they must enter a legally laborious process to terminate parental rights (TPR) before placing the child with a permanent adoptive home. Federal law requires that workers pursue TPR if the child has been out-of-home for twelve consecutive months or fifteen of twenty two consecutive months.

After TPR, if a foster family has cared for the child, the worker typically attempts to place the child with that caregiver. Unfortunately, this may not always be possible. When it is not, the worker undergoes an intensive search for a new permanent home. Many counties use a \textit{family-driven} matching process that strongly resembles a Serial Dictatorship. The caseworker announces that some children require an adoptive home to a pre-specified pool of adoptive homes. Each home sequentially initiates the matching process with the home's most preferred child that still requires a match. Although the order of the dictatorship is first-come-first-serve, we model this as a random order. Relatively fewer counties follow a primarily \textit{caseworker-driven} matching process. The caseworker representing a child searches through specific adoptive homes and reaches out to initiate the match. \qcite{dierks-search-2024} show that, generally, caseworker-driven search is superior to family-driven search, and, in some ways, our results echo this because the county always leads in proposing matches with a centralized clearinghouse.

Our work fits into this final step. Existing match tools do not use deep insights from matching theory. However, a child welfare system cannot maximize placements without considering that utilizing some homes for some children may necessitate \textit{not} utilizing those same homes for other children. A combinatorial problem can drastically harm the number of placements made. Our work on implementing centralization to preempt dynamic incentives is an attempt to address some of these issues.

Once a worker finds a suitable, willing home for a children, in most states, the home tentatively fosters the child for six months before the adoption can be finalized. Placement stability refers to a placement that ends in adoption. The greatest threat in this stage is disruption. If a home decides the child is not the right fit, they can decline to adopt. Social work and pediatric care literature highlight placement stability as a key goal for matches because disruption adversely impacts children's health (\qcitenp{rubin2007impact}). For this reason, most workers' efforts are aimed toward finding matches that will not disrupt. We refer to placement stability as persistence or non-disruption to avoid confusion with the matching-theoretic concept of stability, and we use this as our welfare measure and preferences for children.

%% file: Components/LiteratureContribution.tex
\subsubsection*{Literature Contribution}

The literature on market design for child welfare systems is small. The intersection between matching market design and adoption is even smaller. The first notable work in this space is from \qcite{baccara-2014}. They estimate preferences for private adoptions and find significant heterogeneity. They also provide empirical evidence for discrimination against boys and African American children put up for adoption. Parallels concerns exist in the foster care domain. Their policy recommendations allude to centralizing private adoptions, and we take a similar stance with respect to permanency for foster care children. \qcite{slaugh-2016} take an operations approach to improve the Pennsylvania Adoption Exchange's matching process. The paper provides extensive descriptive and quantitative information about matching in foster care. They simulate counterfactual placements when utilizing a spreadsheet matching tool. Their results indicate that matching tools and combining geographic regions (i.e., inducing more centralization and market thickness) substantially increase placements. Matching tools and technology that specifically guide caseworkers to the right home rather than announcing potential matchings to multiple homes also improve match quality and adoption prospects (\qcitenp{dierks-2024}). This observation is key. When caseworkers announce children in need of permanency to a large pool of homes, organized and coordinated matching becomes nearly impossible. However, a caseworker-driven matching process can facilitate the centralization efforts we propose. 

A few other works, including the aforementioned, pursue dynamic matching approaches to characterize foster care adoption. All of these do so from a \textit{decentralized} perspective. \qcite{mac-donald-2023} explicitly models a fostering process that can transition to adoption. Lower subsidies to foster parents after adopting children in their care distorts incentives away from adopting children with disabilities. Children with disabilities, in general, are much less likely to find non-institutional and permanent homes. \qcite{cortes-2021} develops an empirical matching framework to estimate foster children outcomes under a policy that increases market thickness through delays and coarser regions. Compared to delaying, the decrease in match disruption rates is much larger from coarser regions.

A last group of papers study market design problems tangential to matching in foster care. \qcite{dierks-2024} collaborate with a match recommendation service to study families' misreporting behavior under different placement mechanisms in a static framework. They find that different mechanisms incorporating incentive compatibility come at the cost of placements for low needs and high needs children. \qcite{altinok-2023} design the optimal licensing structure when a regulator needs to screen families that may declare that they can care for needs beyond their capabilities. \qcite{baron-2024} use mechanism design and econometrics machinery to reallocate investigations among Child Protective Services investigators to reduce the number of children unnecessarily sent into foster care and appropriately place the children that are victims of abuse into foster care. 

This paper creates a unique bridge between dynamic matching algorithms and child welfare market design to provide the technical foundation for centralized placements and better placement recommendations in child welfare systems. It may function as a crucial framework for future works studying matching in child welfare systems. Unsurprisingly, our main results show major benefits to centralization as papers like \qcite{cortes-2021} and \qcite{slaugh-2016} suggest. We offer more insights into why decentralized matching tends to lead to systemic inefficiency with theoretical and empirical analysis. Our main results add to the literature through showing that it is possible to provide dynamic incentives for homes to accept placements earlier and sustain truth-telling as a weakly dominant strategy when the matchmaker uses predictive technology to supplement preference elicitation. Furthermore, we demonstrate our mechanisms' robustness to biased and inefficient preference estimators, offering evidence that patience-freeness is effective even when the matchmaker cannot perfectly discern homes' preferred placements. One can see our research as an argument for centralized placements in child welfare systems and as an initial foray into developing tools that practitioners can use immediately.

%% file: Components/Outline.tex
\subsubsection*{Paper Organization}

The rest of this paper proceeds as follows. In the first section, we present our model. In section II, we present the main theoretical results on weakly fair, patience-free, and non-wasteful mechanisms, their existence, and their theoretical guarantees. In section III, we explore our mechanisms' strategic incentives and offer a strategy-proof mechanism. In section IV, we present our empirical simulations. Finally, we conclude in section V with applied insights and discussions on operating and implementing our mechanisms.

%% file: Paper/Model.tex
\section*{I \hspace{10pt} Model}

We outline our model and key primitives below.

\input{Components/Primitives}

\input{Components/Matching}

\input{Components/Stability}

%% file: Components/Primitives.tex
\subsection*{I.A \hspace{10pt} Primitives}


Children arrive to the market (welfare system) over a finite time horizon $T$. A group arrives at every time $t \in \{ 1, 2, ..., T \}$ which we denote as $c(t)$\footnote{For the sake of tractability, we ignore sibling groups in our analysis. We make dynamic sibling group placements the goal of future work.} with an individual child as $c \in c(t)$. A group of homes also arrive at every time $t$ which we denote as $h(t)$ with an individual home $h \in h(t)$. We think of every time period $t$ as one month, but any arbitrary interpretation of time can fit our model. 


A matchmaker would like to assign children to homes. A home $h$ matching with child $c$ receives utility $V_h(c)$. Conversely, a child receives utility $U_c(h)$ from the same match. In the context of child welfare systems, one should interpret a child's utility as the value of the match from the matchmaker's perspective. Following \qcite{baccara-2020}, we include additive waiting costs in our model, where we assume that a common waiting cost $w_c > 0$ incurs for every child that remains unmatched after a time $t$. A child arriving at $t$ has a time-dependent utility at time $k$ equal to $U_c^k(h) = U_c(h) - (k - t)w_c$. Likewise, homes incur a common waiting cost $w_h$ with time-dependent utility $V_h^k(c) = V_h(c) - (k - t)w_h$. Some homes might have negative utility for certain children (for example, a home that never wants to adopt a child older than 13). We say that a child $c$ is \textit{acceptable} to a home $h$ if $V_h(c) \geq 0$, and a home $h$ is acceptable to a child $c$ if $U_c(h) \geq 0$.

\begin{assumption}
    All children and homes have a strict preference ordering, that is, $\forall c$, there does not exist two $h,h'$ with $h \neq h'$ such that $U_c(h) = U_c(h')$, and, $\forall h$, there does not exist two $c,c'$ with $c \neq c'$ such that $V_h(c) = V_h(c')$.
\end{assumption}


Our model will consider the fact that homes may accept or decline placements, specifically, each home $h$ decides to accept or decline a placement offered at time $t$ which is an action $a_h^t \in \{ 0, 1\}$ where $0$ means decline and $1$ means accept. 


Our axioms will implicitly assume that all homes have deterministic knowledge about future arrivals and their types. A home's option value of waiting is endogenously formed by both the expectation of future arrivals and what mechanism would offer. We design matching mechanisms that simultaneously \textit{do not} use the future knowledge that homes possess yet still ensure timely participation. We do not pursue a Bayesian model where homes have priors over future arrivals because this would relax the matchmaker's problem and complicate analysis without significantly changing results. Assuming deterministic knowledge forces frees the mechanism we design to guarantee timely participation for homes under any possible realization of arrivals.

We also assume that the matchmaker has access to children and homes' utilities and waiting costs. Utilities and waiting costs for children are simple to determine as a child's welfare from a match and costs are typically evaluated by the matchmaker. We will show that it is not possible to fully elicit utilities for homes in our strategic analysis. Our approach involves predicting preference intensity after eliciting acceptability. We show our results' responsiveness to prediction error in our empirical section.

%% file: Components/Matching.tex
\subsection*{I.B \hspace{10pt} Dynamic Matching}

The set $C(t)$ contains all child arrivals up to time $t$, and the set $H(t)$ contains all home arrivals up to time $t$:
\begin{align*}
    C(t) = \bigcup_{k=1}^t c(t) \text{ and } H(t) = \bigcup_{k=1}^t h(t)
\end{align*} 
The market at time $t$ is the undirected bipartite graph and preferences
\begin{align*}
    M(t) = (C(t), H(t), E(t), U(t), V(t))
\end{align*}
where $E(t) \equiv \{ \{ c, h \} : c \in C(t), h \in H(t)\}$ is the possible edges (matches) at time $t$, $U(t) = \{ U_c(h) \}_{\forall(c,h):c \in C(t), h \in H(t)}$, and $V(t) = \{ V_h(c) \}_{\forall (h,c) : h \in H(t), c \in C(t)}$. An instance of a dynamic matching problem is a sequence of markets $M^T = \{ M(t)\}_{1 \leq t \leq T}$. The space of all problems (sequences of markets) at time $t$ is $\mathcal{M}^t$.

An admissible matching $\mu \subseteq E(t)$ at time $t$ selects one-to-one matchings, i.e., on $\mu$, each vertex in $C(t)$ and $H(t)$ only have one incident edge. We write that $c$ and $h$ are matched if $\{ c, h \} \in \mu$. Alternatively, we denote this as $\mu(c) = h$ and $\mu(h) = c$. Say that $\mu(c) = c$ if $c$ has no incident edge on $\mu$ and likewise for $h$. From hereon, we refer to an admissible matching as a matching.

Matchmakers may propose any placements they like, but homes must accept them for the match to proceed. Define $a_h(t) \equiv (a_h^k)_{1 \leq k \leq t}$ and $a(t) \equiv (a_h(t))_{\forall h}$. We call $A(t)$ the action space at $t$ where $a(t) \in A(t)$. A market matching $m_t = (a(t), \mu_1, \mu_2, ..., \mu_t)$ is a tuple consisting of actions and time-consistent matchings satisfying (1) for any agent $i$ such that $\mu_k(i) \neq i$ and $a^k_h = 1$ ($a^k_{\mu_k(i)} = 1$), we have that $\mu_j(i) = i$ for any $j > k$ and (2) $\mu_k \in E(k)$. Time-consistency implies that matchings are irreversible once accepted, which we impose through excluding matched agents in future time periods, and that the matching can only include agents that have arrived. The space of all market matchings at time $t$ is $\mathcal{G}^t$

The Revelation Principle implies that actions are unnecessary in the model. However, our analysis relies on home conjectures of what matchings might ensue if they were not to match in some time periods. It is possible to define axioms that implicitly consider such deviations, but it is considerably more difficult to describe mechanisms that condition on home actions---as ours do---that are left implicit. Furthermore, we simulate the extensive form games where homes can accept or decline placements. For these reasons, we explicitly write homes' actions.

A \textbf{dynamic mechanism} is a (deterministic) sequence of rules $q = \{q_t\}_{1 \leq t \leq T}$ creating matchings at every $t$ on the market at $t$ which may depend on the history. Formally, $q_t : \mathcal{M}^t \times \mathcal{G}^t \rightarrow E(t)$. A rule $q_t$ is feasible if its matchings are time consistent with the market history. A market matching $m_{t-1}$ endogenizes the sets of children and homes available to match at $t$ which we denote:
\begin{align*}
    C(t|m_{t-1}) &\equiv \{ c : c \in C(t) \text{ and, for all } k \leq t-1 \text{, either } \mu_k(c) = c \text{ or } a_{\mu_k(c)}^k = 0 \} \\
    H(t|m_{t-1}) &\equiv \{ h : h \in H(t) \text{ and, for all } k \leq t-1 \text{, either } \mu_k(h) = h \text{ or } a_h^k = 0\}
\end{align*}
A child remains in the market if and only if no home has adopted her up till the present, and a home remains if and only if it makes no adoptions until the present. 

Every period that a home is unmatched, it incurs a waiting cost $w_h$, and every period a child is unmatched, she incurs a waiting cost $w_c$. Let $t_{i|m}$ denote the earliest time period for any agent $i$ such that $\mu_t(i) \neq i$. A home arriving in period $k$ receives payoff $V_h(\mu_t(h)) - (t_{h|m} - k) w_h$ or $(T - k) w_h$ if the home is never matched, and the matchmaker receives payoff $U_c(\mu_t(c)) - (t_{c|m} - k) w_c$ or $(T - k)w_c$ if $c$ is never matched. The matchmaker incurs the aggregate waiting costs for children:

\begin{center}
    $W_{C|m} \equiv \sum_{t=1}^T \sum_{c} \mathbbm{1}\{ c \in C(t|m_{t-1}) \} w_c$
\end{center}


While our primary focus is the design of the dynamic mechanism, we specify the timing of the extensive form game induced on the mechanism below:

\begin{enumerate}
    \item The matchmaker announces $Q$.
    \item At period $t$, each home $h$ receives a placement according to $Q$.
    \item Each home $h$ simultaneously decides to accept or decline the placement.
    \item $t$ ends.
\end{enumerate}

This continues until $t = T$, at which point the game ends at (4).

%% file: Components/Stability.tex
\subsection*{I.C \hspace{10pt} Fairness and Efficiency}


Our motivation is developing axioms that fit our applied setting that also provides useful insights for subsequent theoretical work. Two key observations are that, first, if a matchmaker can commit to contingent matchings that are sequentially inefficient, subject to axiomatic limitations on this inefficiency, then the matchmaker can motivate timely matches. Second, requiring that participant blocking pairs are consistent with possible counterfactual assignments of the mechanism limits deviations. These two allowances imply that we do not need to pursue a notion as strict as those appearing in the literature.

Our solution is patience-freeness. First, we review familiar properties in matching markets. The following properties hold for a static set of children $C$ and static set of homes $H$.

A child $c \in C$ and home $h \in H$ are a blocking pair on $\mu$ if $U_c(h) > U_c(\mu(c))$ and $V_h(c) > V_h(\mu(h))$. $\mu$ is \textit{individually rational} on $C$ and $H$ if, for all $c \in C, h \in H$, $U_c(\mu(c)) \geq 0$ and $V_h(\mu(h)) \geq 0$. $\mu$ is \textit{fair} on $C$ and $H$ if it has no blocking pairs and is individually rational. We say that a blocking pair $c$ and $h$ also have \textit{justified envy} if $h$ has an assignment, i.e., $\mu(h) \neq h$. A matching $\mu$ is \textit{weakly fair} on $C$ and $H$ if it has no justified envy and is individually rational.

Eliminating justified envy is less demanding than preventing blocking pairs. We do not count blocking pairs where the home is unmatched. Together with our other axioms, an unmatched home's envy will be "unjustified" because the matchmaker only excludes the home to induce timely participation. Fairness is critical in this context to respect welfare for children and to respect the adoptive homes' preferences.

Our main axiom prioritizes the elimination of dynamic incentives for homes to decline placements. A matching at period $t$ for a mechanism $q$ given a history $m_{t-1}$ is $\mu_{t,m}^q = q(m_{t-1}, M^t)$. A matching $\mu_t \in m_t$ has a waiting option if there exists some $h$ such that $V_h^t(\mu_t(h)) > V_h^k(\mu_k(h))$ for some $k < t$ where $\mu_k \in m_t$ and $\mu_k(h) \neq h$. $\mu_t$ is \textit{patience-free} if it has no waiting options. We refer to the class of weakly fair and patience-free matchings at $t$ given $m_{t-1}$ as $\phi_{t,m}$. A mechanism $q$ has a given property if its matching $\mu_{t,m}^q$ has the property at every $m_{t-1}$. 

A home contemplating rejecting a placement does not "go outside" of the market to attempt to form an intertemporal blocking pair. Rather, the home reasons based on what the mechanism would deliver it if it instead participates at a later time given its knowledge of the market's evolution---this is a waiting option\footnote{Child welfare systems fit this description. The matchmaker dictates the child's placement. When a home decides not to accept a placement, it is rarely, if ever, possible for that home to foster or adopt another child that the matchmaker does not consent to.}. It accepts the present match if it is better than the future time-discounted match.

A dynamic mechanism $q$ is \textit{weakly non-wasteful} if, given that $m_{t-1}$ satisfies $a_h^k = 1$ for all $k \leq t - 1$ and $h \in H(t)$, there does not exist a $c \in C(t|m_{t-1})$ and $h \in H(t|m_{t-1})$ such that $U_c(h) \geq 0$, $V_h(c) \geq 0$, $\mu^q_{t,m}(c) = c$, and $\mu^q_{t,m}(h) = h$. $q$ is strictly non-wasteful if this holds for all $m_{t-1}$\footnote{It is immediate that a fair mechanism is strictly non-wasteful.}.

When a mechanism is (weakly) non-wasteful\footnote{We omit "weakly" except where necessary to ease exposition.}, it will always place all possible children and homes when homes accept the matchmaker's placements. Non-wastefulness is a key property that we exploit to allow the matchmaker to shape homes' dynamic incentives toward accepting the matchmaker's offers. We also use the fact that it is a weak property to dispense far-reaching conclusions for strategy-proof mechanisms (Proposition \ref{pr:noexistsp}). Nevertheless, weak non-wastefulness leaves room for mechanisms that can have excess unfilled capacity, for example, a mechanism that ends the market after any home rejects a placement. 

We introduce one last comparative axiom. A matching $\mu$ has one-sided unfilled capacity $\kappa(\mu|H) = \{ h \in H : \mu(h) = h \}$. A matching $\mu$ is capacity dominated by $\mu'$ on $H$ if $\kappa(\mu'|H) \subset \kappa(\mu|H)$. $\mu$ is $\Omega$-capacity efficient on $H$ if it is not capacity dominated by any $\mu \in \Omega$ on $H$. A mechanism $q$ is capacity efficient if $\mu_{t,m}^q$ is $\phi_{t,m}$-capacity efficient on $H(t|m_{t-1})$ for any $m_{t-1}$. $q$ satisfies this if it is not capacity dominated by a weakly fair and patience-free matching.

%% file: Paper/DynamicEnvyFree.tex
\section*{II \hspace{10pt} Weakly Fair and Patience-Free Mechanisms}

In this section, we first offer key existence and non-existence proofs. We also highlight how a multiplicity of mechanisms warrants further exploration. Last, we explore important implications of patience-freeness that guarantee expedient, voluntary participation among homes.

\input{Components/Existence}

\input{Components/Participation}

%% file: Components/Existence.tex
\subsection*{II.A \hspace{10pt} Existence}

Our primary efficiency measure is unfilled capacity. Our initial results show that there is tension between fairness, patience, and capacity.

\begin{proposition} \label{pr:noexistence}
    A fair and patience-free mechanism does not exist.
\end{proposition}

\input{Components/Proofs/NoExistenceSNW}



\subsubsection*{Home Penalized Deferred Acceptance}

Here, we introduce our first mechanism. It is a modification of sequential Deferred Acceptance (DA) that incorporates previous placements and homes' decisions to reject those placements into present match decisions. \\~\

\noindent\textbf{Step 0}: At each $t$, initialize each child $c$'s and each home $h$'s preferences as $\Tilde{U}_c(\cdot) = U_c(\cdot)$ and $\Tilde{V}_h(\cdot) = V_h(\cdot)$, respectively.

\noindent\textbf{Step 1:} if $t > 1$, for each $h \in H(t|m_{t-1})$, if $\mu_k(h) \neq h$ for some $k < t$, then
\begin{itemize}
    \item Let $k$ be the most recent period such that $\mu_k(h) \neq h$ and $\mu_k \in m_{t-1}$. Set $p_h^t = 1$ if there exists a $c \in C(t|m_{t-1})$ such that $V_h^t(c) > V_h^k(\mu_k(h))$ and $ U_{c}(h) \geq 0$. Set $p_h^t = 0$ otherwise.
    \item If $p_h^t = 1$, then set $\Tilde{V}_h(c') < 0$ for all $c'$.
\end{itemize}

\noindent\textbf{Step 2:} Using $\Tilde{U}_c(\cdot)$ and $\Tilde{V}_h(\cdot)$ as preferences
\begin{itemize}
    \item Each $h \in H(t|m_{t-1})$ without a child holding its proposal proposes to its best, acceptable $c \in C(t|m_{t-1})$ that it has not yet proposed to, if it can propose to any.
    \item Each $c \in C(t|m_{t-1})$ holds her best, acceptable proposal and rejects all others.
    \item Repeat 1-2 until no additional proposals are made.
\end{itemize}
\noindent\textbf{Step 3:} Set $\mu_t(c)$ to be each child's held proposal; $\mu_t(c) = c$ otherwise. Set $\mu_t(h)$ to be the child holding a home's proposal; $\mu_t(h) = h$ otherwise. \\~\

\begin{theorem} \label{th:HPDA}
    Home Penalized Deferred Acceptance is weakly fair, patience-free, and non-wasteful.
\end{theorem}

\begin{proof}
    See Appendix A.
\end{proof}

Since home-proposing DA yields a home-optimal matching, there cannot be any justified envy, which implies weak fairness. HPDA is patience-free because regardless of what placements homes have accepted or rejected, the fact that a home $h$ is active at time $t$ indicates that $h$ has received no placements or has rejected all previous placements. The matchmaker knows which case it is. In the latter, the matchmaker can delay $h$'s placements. HPDA dynamically adjusts the relative payoff between past rejected placements and present placements through delaying homes in the matching process.




\begin{corollary} \label{cr:existence}
    A weakly fair, patience-free, and non-wasteful mechanism always exists.
\end{corollary}

\begin{proof}
    This follows directly from HPDA's existence.
\end{proof}

However, HPDA is not unique in this class of mechanisms. Proposition \ref{pr:noexistence} suggests that other mechanisms satisfying these properties must operate similarly to HPDA, imposing delays on non-compliant homes. Nevertheless, HPDA's rule is too naive. It delays a home if there is \textit{any} child that would be a waiting option, even if that child is unattainable in a weakly fair matching. 

\subsubsection*{Ascending Deferred Acceptance}

We show that there is a smarter, more capacity-efficient method to penalize homes that maintains weak fairness and patience-freeness. We introduce the Ascending Deferred Acceptance algorithm (ADA). \\~\

\noindent\textbf{Step 0}: At each $t$, initialize each child $c$'s and each home $h$'s preferences as $\tilde{U}^1_c(\cdot) = U_c(\cdot)$ and $\tilde{V}^1_h(\cdot) = V_h(\cdot)$, respectively.

\noindent\textbf{Step 1:} for any round $n \geq 1$
\begin{itemize}
    \item \noindent\textbf{Step $n$.1:} Using $\tilde{U}^n_c(\cdot)$ and $\tilde{V}^n_h(\cdot)$ as preferences
    \begin{itemize}
        \item Each $c \in C(t|m_{t-1})$ without a home holding its proposal proposes to its best, acceptable $h \in H(t|m_{t-1})$ that it has not yet proposed to, if it can propose to any.
        \item Each $h \in H(t|m_{t-1})$ holds her best, acceptable proposal and rejects all others.
        \item Repeat until no additional proposals are made.
    \end{itemize}

    \item \noindent\textbf{Step $n$.2:} For any $h$, let $c$ be the proposal that $h$ is holding. Let $k$ be the most recent period such that $\mu_k(h) \neq h$ and $\mu_k \in m_{t-1}$, if any exists. If $V_h^t(c) > V_h^k(\mu_k(h))$, then set $\tilde{V}^{n+1}_h(c') < 0$ for all $c' \in C(t)$. Otherwise, set $\tilde{V}^{n+1}_h(c') = \tilde{V}^n_h(c')$. Exit to Step 3 if $\tilde{V}^{n+1}_h(c') = \tilde{V}^n_h(c')$ for all $h$. Otherwise, continue to round $n+1$.
\end{itemize}

\noindent\textbf{Step 3:} Set $\mu_t(c)$ to be each child's held proposal; $\mu_t(c) = c$ otherwise. Set $\mu_t(h)$ to be the child holding a home's proposal; $\mu_t(h) = h$ otherwise. \\~\

\begin{theorem}\label{th:ADA}
    Ascending Deferred Acceptance is weakly fair, patience-free, non-wasteful, and capacity efficient.
\end{theorem}
\begin{proof}
    See Appendix A.
\end{proof}

We refer to ADA as \textit{ascending} because the properties of Deferred Acceptance imply that each time ADA runs child-proposing DA, it matches each home its least favorite partner in any weakly fair matching where the matched children and homes are matched. This implies that a home with a waiting option cannot be matched in \textit{any} weakly fair matching given the other matched children and homes. As ADA sequentially removes homes in each repetition, each remaining home becomes weakly better off after re-running DA, maintaining weak fairness. These properties yield the result and underwrite the rationale for utilizing child-proposing rather than home-proposing DA. ADA's capacity efficiency implies that there does not exist penalties that match more homes (in a set inclusion sense) satisfy both weak fairness and patience-freeness. Note that because DA runs in polynomial time and at least one home is removed from the ADA algorithm in each round $n$, ADA also runs in polynomial time implying that it is computationally tractable.

%% file: Components/Proofs/NoExistenceSNW.tex
\begin{proof}
    We demonstrate Proposition 1 via. counterexample. We suppress $q,m$ in super- and subscripts where it is obvious from context. Consider an environment where $T = 2$, $h(1) = \{h_1\}$, $c(1) = \{c_1\}$, $h(2) = {h_2}$, and $c(2) = \{c_2\}$.

    \begin{center}
        \begin{tikzpicture}[
            roundnode/.style={circle, draw=blue!80, fill=blue!70, very thick, minimum size=6mm},
            squarednode/.style={rectangle, draw=blue!50, fill=blue!40, very thick, minimum size=7mm},
            empty/.style={rectangle, draw=white, fill=white, minimum size = 5mm},
            ]
                

            \node[roundnode] (c1) {$c_1$};
            \node[squarednode] (h1) [below=of c1] {$h_1$};
            \node[empty] (t1) [above=0.5 of c1] {$t = 1$};


            \node[roundnode] (c2) [right=of c1, draw=red!80, fill=red!70] {$c_2$};
            \node[squarednode] (h1) [below=of c2] {$h_2$};
            \node[empty] (t1) [above=0.5 of c2] {$t = 2$};
        \end{tikzpicture}
    \end{center}

    Players' preferences are:
    \begin{align*}
        V_{h_1}(c_2) = 2 &> V_{h_1}(c_1) = 1 \\
        V_{h_2}(c_1) = 2 &> V_{h_2}(c_2) = 1 \\
        U_{c_1}(h_2) = 3/2 &> U_{c_1}(h_1) = 1 \\
        U_{c_2}(h_1) = 3/2 &> U_{c_2}(h_2) = 1
    \end{align*}
    with $w_c = 2$ and $w_h = 1/2$. Suppose that $q$ is a fair mechanism. Then, it is also strictly non-wasteful. For any $a_h^1$, it must specify $\mu_1(h) = c_1$. Suppose that $a_h^1 = 0$, then $C(2|m_1) = \{c_1, c_2\}$ and $H(2|m_1) = \{h_1, h_2\}$. By strict non-wastefulness, $\mu_2(h_1) \in \{c_1, c_2\}$ and $\mu_2(h_2) \in \{c_1, c_2\}$. If $q$ is fair, $\mu_2$ must satisfy weak fairness, hence $\mu_2(h_1) = c_2$ and $\mu_2(h_2) = c_1$. Since every agent has positive utility for any match, this is also individually rational. However, $q$ must also satisfy patience-freeness. Patience-freeness implies that $V_{h_1}^1(\mu_1(h_1)) = 1 \geq V_{h_1}^2(\mu_2(h_1)) = 2 - 1/2 = 3/2$. This is a contradiction. Therefore, $q$ cannot be both fair and patience-free.
\end{proof}

Specifically, weak fairness, strict non-wastefulness, and patience-freeness are incompatible. Furthermore, a matchmaker would prefer to place $h_1$ with $c_1$ and $h_2$ with $c_2$. These placements would yield payoff $2$ to the matchmaker, whereas waiting until $t=2$ to match $h_1$ with $c_2$ and $h_2$ with $c_1$ yields payoff $1$ because of $c_1$'s waiting cost. The home's patience combined with strict non-wastefulness renders the matchmaker's first-best outcome impossible.

%% file: Components/Participation.tex
\subsection*{II.B \hspace{10pt} Participation}

In our interviews with a U.S. county, workers representing children express serious dismay over managing caseloads where it seems impossible to find the right home to accept a placement for a child, especially when the child has serious needs or disabilities. Our interviewees report that there is significant variation in homes' waiting times and that many homes do wait in hopes of achieving an ideal placement. Our theoretical result below shows that when a matchmaker implements a patience-free mechanism, homes do not have an incentive to reject placements.

\begin{proposition} \label{pr:dominant}
    In the extensive form game on a patience-free mechanism $q$, accepting the first placement is a weakly dominant strategy for all homes.
\end{proposition}

\begin{proof}
    See Appendix A.
\end{proof}

The proposition is immediate given the definition of patience-freeness. One implicit assumption of our model is critical to participation. The matchmaker must be able to limit homes' dynamic incentives through controlling the placement offers that a home receives. In most child welfare systems, homes register with a provider agency responsible for representing the home to the local county. The local county writes contracts with the local provider agencies that specify a number of placements for children that the providers must guarantee to the county, meaning that the provider must reserve a certain number of homes for the local county. The county we interviewed does not pursue placements out-of-county unless the prospective home is kin to the child. These features imply that the matchmaker possesses power over the placement offers for at least some homes residing in the local county. Even if homes occasionally receive outside options, we conjecture that patience-free mechanisms should maintain their efficacy as long as these options are infrequent or unattractive relative to the local authorities' placement offers.

%% file: Paper/StrategyProofness.tex
\section*{III \hspace{10pt} Strategic Incentives}

In this section, we explore homes' strategic incentives under several mechanisms. We first propose an augmented version of strategy-proofness that considers limited misreports and give rationales. Next, we construct benchmark strategic properties using sequential home-proposing DA. Last, we give strategic properties for HPDA, ADA, and a new mechanism: Home Endowed Deferred Acceptance (HEDA).

\input{Components/Manipulations}

\input{Components/StrategicBenchmark}
\input{Components/StrategicProperties}

%% file: Components/Manipulations.tex
\subsection*{III.A \hspace{10pt} Manipulability}

Practitioners report at least one major concern regarding manipulations: homes that report unwillingness to adopt \textit{any} child that does not meet stringent characteristics. In practice, these homes wait to accept placements for children that do not ultimately meet their original preference reports. In our interviews, caseworkers say that homes often broaden their preferences after experience with the system, suggesting that homes could be attempting to manipulate matchmaking through truncating their preferences. Another plausible explanation is that homes simply have imperfect information and learn over time. We believe that both mechanisms may be at work, and eliminating manipulability has promising benefits.

\begin{figure}
    \centering
    \includegraphics[width=.85\linewidth]{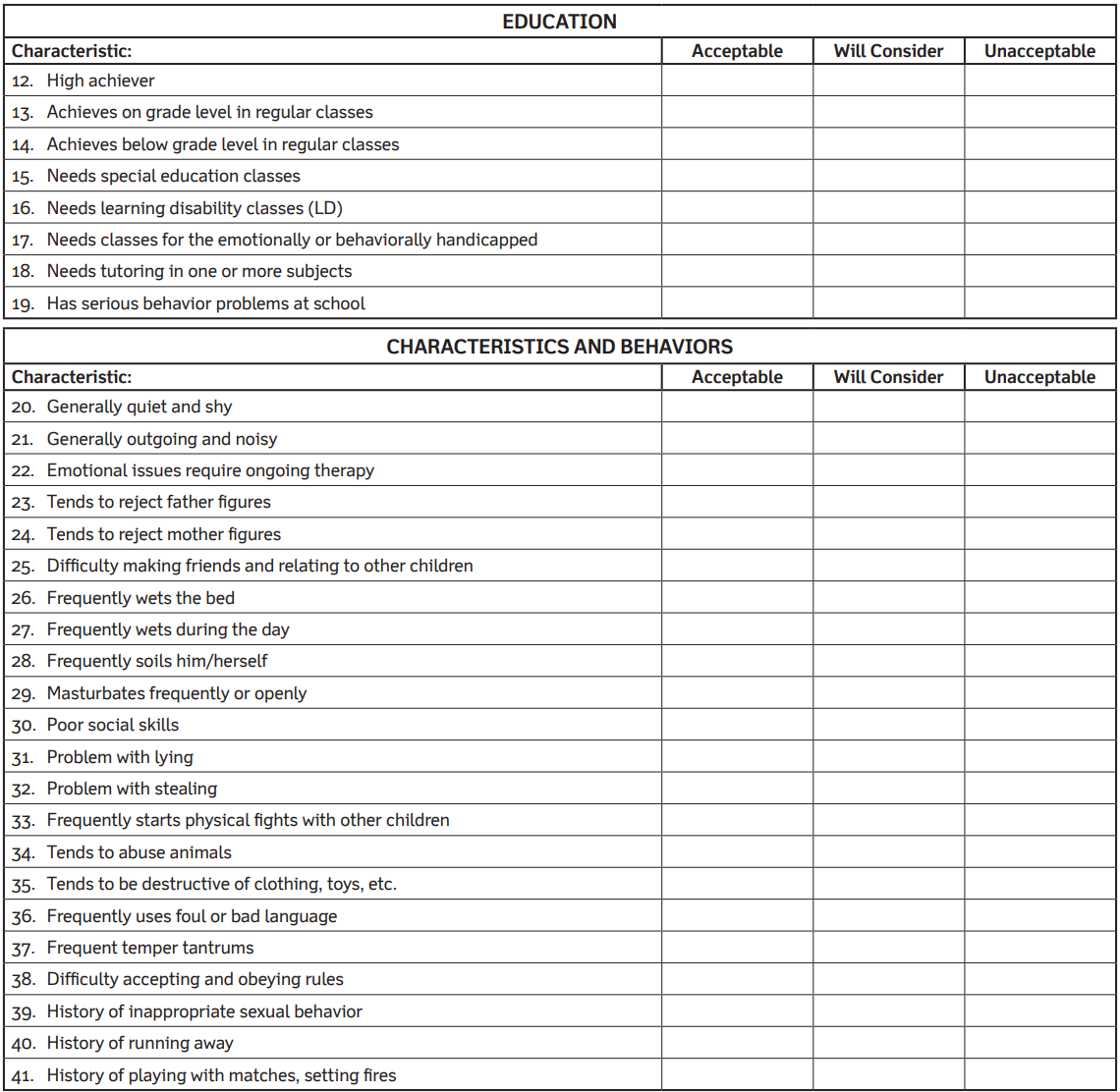}
    \caption{An example of a preference form. We treat the "Will Consider" category as "possibly acceptable".}
    \label{fig:preferences}
\end{figure}

The process to report and use preferences in matchmaking between children and homes is non-standardized across the United States. In the state where the county we interviewed resides, the statewide adoption match recommendation service requests that homes seeking a match fill a form that indicates preferences over children's characteristics (Figure \ref{fig:preferences}).

Keeping in trend with practice, we assume that a home reports acceptability for each child available. Formally, a home's preference report is a strategy $\forall h,c$:
\begin{center}
    $\sigma_h(c) = \begin{cases}
        1 &\text{ if } c \text{ is acceptable to } h \\
        0 &\text{ otherwise }
    \end{cases}$
\end{center}
Denote $\sigma_h \equiv (\sigma_h(c))_{\forall c}$, $\sigma \equiv (\sigma_h)_{\forall h}$, $\sigma_{-h} \equiv (\sigma_{h'})_{\forall h' \neq h}$, and $a_{-h}(t) \equiv (a_{h'}(t))_{\forall h' \neq h}$. The matchmaker observes $h$'s true utility for $c$ if $h$ reports that $c$ is acceptable. Otherwise, the matchmaker can only observe that $c$ is unacceptable to $h$. Preserving continuity in notation, we denote $V_h(c)$ as the utility that the matchmaker observes and $\Hat{V}_h(c)$ as $h$'s true utility for $c$. Then, we have:
\begin{center}
    $V_h(c) = \begin{cases}
        \Hat{V}_h(c) &\text{ if } \sigma_h(c) = 1 \\
        \underline{v} &\text{ otherwise }
    \end{cases}$
\end{center}
We maintain the assumption that $w_h$ is known. The matchmaker may calibrate it based on their own beliefs. Additionally, in our results, we note a variant of our strategy-proof mechanism that maintains its properties even when $w_h$ is unknown.

\begin{assumption} \label{as:nounacceptable}
    If $\Hat{V}_h(c) < 0$, then $\Hat{V}_h(c) = \underline{v}$
\end{assumption}

We set $\underline{v}$ arbitrarily low enough that a home will not accept unacceptable placements to avoid large waiting costs. This would never happen. Rather than incur further waiting costs, a home would leave the child welfare system. Assumption \ref{as:nounacceptable} accommodates this outside option. Since we only examine individually rational mechanisms, the matchmaker will not offer unacceptable placements, and this assumption is without loss of generality.

The market matching
\[m_{t,a(t),\sigma}^q = (m_{t-1,a(t-1),\sigma}^q, a(t-1), q(m_{t-1,a(t-1),\sigma}^q, M^t(\sigma)))\]
where $m_{1,a(1),\sigma}^q = q(M^1(\sigma))$ and $M^t(\sigma) = (C(t), H(t), E(t), U(t), V(t|\sigma))$ now implies that observed utilities for homes $V(t|\sigma)$ are endogenous to $\sigma$. We require the entire market matching history to be determined by $q$ to analyze strategy-proofness properties because actions and reports will endogenize realized utility. A report $\sigma_h$ is \textit{truthful} if and only if $\sigma_h(c) = 1$ when $ \Hat{V}_h(c) \geq 0$ and $\sigma_h(c) = 0$ when $\Hat{V}_h(c) < 0$. 

For the following definitions, fix a mechanism $q$, action profile $a(T)$, and reports $\sigma$. Define $t_h \equiv t_{h|m_{t,a(T),\sigma}^q}$. A home $h$ arriving at $k$ has realized utility $\Hat{V}_h^r(\mathbf{a}(T), \sigma|q) = \Hat{V}_h(\mu_{t_h}(h)) - (t_h - k)w_h$.

An action and report $(a_h(T),\sigma_h)$ are \textit{weakly dominant} for $h$ under a mechanism $q$ if
\[\Hat{V}_h^r(a_h(T), a_{-h}(T), \sigma_h, \sigma_{-h}|q) \geq \Hat{V}_h^r(a'_h(T), a_{-h}(T), \sigma'_h, \sigma_{-h}|q)\]
for all $a'_h(T), a_{-h}(T), \sigma'_h, \sigma_{-h}$. A mechanism $q$ is \textit{strategy-proof} if, for all $h$, $(a_h(T), \sigma_h)$ is weakly dominant, where $\sigma_h$ is truthful and $a_h^t = 1 \hspace{5pt} \forall t$.

Matchmakers want all homes want to report that acceptable children are acceptable and unacceptable children are unacceptable. Avoiding placing a child with a home that would not accept the child saves valuable time and resources, and correctly discerning what children are acceptable to homes assists in matching, recruitment, and planning. 

We change the timing of the game as follows:

\begin{enumerate}
    \item The matchmaker announces $q$.
    \item Each home $h$ reports $\sigma_h(c)$ for all $c$.
    \item At period $t$, each home $h$ receives a placement according to $q$ and reported preferences.
    \item Each home $h$ simultaneously decides to accept or decline the placement.
    \item $t$ ends.
\end{enumerate}

This continues until $t = T$, at which point the game ends at (5).

%% file: Components/StrategicBenchmark.tex
\subsection*{III.B \hspace{10pt} Benchmark}
In our empirical section, we  perform two comparisons. The first is patience-free placements contra decentralized placements which we model using a Random Serial Dictatorship (RSD). The second is patience-free placements contra sequentially stable placements. RSD does not guarantee any desirable properties in this setting, therefore, we focus on the second comparison in our theoretical analysis. Here, we introduce sequential Deferred Acceptance (SDA) which is a spot placement mechanism that indiscriminately runs home-proposing or child-proposing DA. In our main specifications, we use home-proposing DA. \\~\

\noindent\textbf{Step 1:} $t \geq 1$: Using $U_c(\cdot)$ and $V_h(\cdot)$ as preferences
    \begin{itemize}
        \item Each $h \in H(t|m_{t-1})$ without a child holding its proposal proposes to its best, acceptable $c \in C(t|m_{t-1})$ that it has not yet proposed to, if it can propose to any.
        \item Each $c \in C(t|m_{t-1})$ holds her best, acceptable proposal and rejects all others.
        \item Repeat 1-2 until no additional proposals are made.
        \item Set $\mu_t(c)$ to be each child's held proposal; $\mu_t(c) = c$ otherwise. Set $\mu_t(h)$ to be the child holding a home's proposal; $\mu_t(h) = h$ otherwise. \\~\
    \end{itemize}

The proposition below offers insight on theoretical differences between SDA's properties and our patience-free mechanisms.

\begin{proposition} \label{pr:DAprops}
    Sequential Deferred Acceptance is fair, but it is neither patience-free nor strategy-proof.
\end{proposition}

\begin{proof}
    SDA does not condition on $m_{t-1}$, and it is fair at every $\mu_t$ by usual results. Since it is fair, by Proposition \ref{pr:noexistence}, it cannot be patience-free. Fairness implies strict non-wastefulness which implies that sequential DA is non-wasteful. Then, by Proposition \ref{pr:noexistsp} below, it cannot be strategy-proof.
\end{proof}

A home that expects very good placements in the future but does not want to reject prospective adoptive children can declare unacceptability for large swathes of children to avoid less desirable placements. Moreover, this alters the matchmaker's consideration set for some placements and may even, from the home's perspective, save the matchmaker from unnecessary effort costs. 

On the contrary, if the home truthfully reports, then the matchmaker might offer placements that the home would find optimal to reject. This blocks the rejected child from other placements that might be accepted in that same time period. This might also motivate homes to use preference reports to deter competition. Reporting some children as unacceptable implies that they might match to other homes and cause those homes to exit the market earlier. Preference misreports and rejections have meaningful outcome differences under SDA, and it is not clear that there is a strong incentive for homes to report truthfully.

%% file: Components/StrategicProperties.tex
\subsection*{III.C \hspace{10pt} Patience-Free and Strategy-Proof Mechanisms}

We turn our attention toward examining the strategic properties of our mechanisms. Our first result dispenses immediate conclusions for HPDA.

\begin{proposition} \label{pr:noexistsp}
    A weakly fair, non-wasteful and strategy-proof mechanism does not exist.
\end{proposition}

\input{Components/Proofs/SP-NW}

We gave the intuition in the benchmark. Given the broad scope of this result---and that it persists if one replaces 'weakly fair' with a rudimentary notion of Pareto efficiency---we suggest that preference truncations could be a potential mechanism explaining child welfare systems' perpetual lack of homes to care for older children and children with higher needs, and strategy-proof mechanisms could be a solution.


\begin{corollary}
    Home Penalized Deferred Acceptance and Ascending Deferred Acceptance are not strategy-proof.
\end{corollary}

\begin{proof}
    Both are weakly fair and non-wasteful. By Proposition \ref{pr:noexistsp}, neither are strategy-proof.
\end{proof}

The key manipulation that exploits HPDA is the same as in the above counterexample, but homes have even stronger incentives to misreport. Since HPDA is weakly fair and patience-free, it cannot allow $h$ to match with $c_2$ if $h$ rejects $c_1$ at $t = 1$. Hence, $h$ can only match with $c_2$ if it reports that $c_1$ is unacceptable. 

These results temper overly optimistic expectations of any non-wasteful mechanism's performance and motivate a search for a mechanism that is strategy-proof. We develop this in our next mechanism. Home Endowed Deferred Acceptance (HEDA) fixes a range of utility that a home can receive at any time period relative to its entry into the child welfare system; i.e., each home receives an endowment at each time period. Following this, we run home-proposing DA with proposals restricted to children within each home's endowment. \\~\

\textbf{Step 0:} Fix some arbitrary maximum and minimum utilities $\Bar{V}$ and $\ubar{V}$. Fix an arbitrary set of $N$ disjoint, compact intervals $(E_i)_{0 \leq i \leq N}$ spanning $[\ubar{V}, \Bar{V}]$ where the minimal element in $E_i$ is greater than the maximal element in $E_{i+1}$. Home $h$ arriving in period $k$ has a time-dependent endowment $B_t(h) \equiv E_{t-k}$ if $t - k \leq N$ and $B_t(h) \equiv \{\ubar{V}\}$ otherwise. 

\textbf{Step 1:} For any $t \geq 1$, initialize $\Tilde{V}_h(\cdot) = V_h(\cdot)$.
\begin{enumerate}
    \item Set $e_h^t(c) = 1$ if $V_h^t(c) \in B_t(h)$ and $V_h(c) \geq 0$. $e_h^t(c) = 0$ otherwise.
    \item If $e_h^t(c) = 0$, then set $\Tilde{V}_h(c) < 0$.
\end{enumerate}
    
\textbf{Step 2:} For any $t \geq 1$, using $U_c(\cdot)$ and $\Tilde{V}_h(\cdot)$ as preferences
\begin{enumerate}
    \item Each $h \in H(t|m_{t-1})$ without a child holding its proposal proposes to its best, acceptable $c \in C(t|m_{t-1})$ that it has not yet proposed to, if it can propose to any.
    \item Each $c \in C(t|m_{t-1})$ holds her best, acceptable proposal and rejects all others.
    \item Repeat 1-2 until no additional proposals are made.
    \item Set $\mu_t(c)$ to be each child's held proposal; $\mu_t(c) = c$ otherwise. Set $\mu_t(h)$ to be the child holding a home's proposal; $\mu_t(h) = h$ otherwise. \\~\
\end{enumerate}
We show that HEDA is strategy-proof, but it is not weakly fair. We hypothesize that it is possible to derive a patience-free, weakly fair, and strategy-proof mechanism, but it would exhibit serious deficiencies in wastefulness that defeat the purpose of fairness. We test the performance of HEDA in section IV.

\begin{theorem} \label{th:HEDA}
    Home Endowed Deferred Acceptance is patience-free, individually rational, and strategy-proof.
\end{theorem}

\begin{proof}
    See Appendix A.
\end{proof}

The algorithm's success rests on the endowment's invariance to homes' strategies. Because a home always receives its best placement at the earliest period relative to its arrival, a home's optimal strategy is to accept the earliest placement. HEDA's strategy-proofness is non-trivial. In the first period where a home's misreport would cause any match for any home to differ compared to when the home is truthful, we prove that it must be that the misreporting home receives the same match under misreporting/truth-telling or else the misreporting home's match differs. In the former case, the home should accept the placement by the same logic that we use to show HEDA is patience-free; waiting until later yields a worse placement. Misreporting cannot increase the home's utility since the first placement is the same as under truth-telling. The latter case reduces to demonstrating that HEDA is strategy-proof in a static sense; we show this using a novel, amended proof strategy based on \qcite{roth-2017}. Misreporting cannot improve the home's utility even when the home's match changes. Tying the two cases together yields the theorem.

HEDA is not weakly fair because it partitions homes' potential placements using endowments that are ex-ante fixed. Therefore, it is possible for a home and child to mutually benefit from an endowment violation. Moreover, HEDA is not non-wasteful for the same reason. Even a compliant home might not receive a placement at $t$ if no children exist within its endowment, and some other acceptable children outside of that home's endowment might not be placed elsewhere. These two features are the primary downsides of HEDA. 

Proposition \ref{pr:noexistsp} cautions against pursuing a weakly fair, patience-free, and strategy-proof mechanism. Patience-freeness and non-wastefulness imply that any unfilled capacity exists to penalize noncompliant homes. In contrast, a patience-free mechanism without the binding restrictions of non-wastefulness---as a weakly fair, patience-free, and strategy-proof mechanism must be---may be arbitrarily unfair. 

HEDA's structure is amenable to another useful alteration. We have assumed that the matchmaker has access to a reliable calibration for $w_h$, but this may not hold for every county. We propose HEDA* when $w_h$ is unknown: it sets a home's endowment $e_h^t(c) = 1$ if $V_h(c) \in B_t(h)$ rather than relying on the home's time-discounted utility. HEDA* does not differ theoretically from HEDA so Theorem \ref{th:HEDA} still applies\footnote{HEDA is equivalent to running HEDA* with the same endowment interval bounds shifted upward by $(t-k)w_h$.}. 

%% file: Components/Proofs/SP-NW.tex
\begin{proof}
    We prove this by counterexample. $T = 2$, $h(1) = \{h\}$, $c(1) = \{c_1\}$, $h(2) = \varnothing$, and $c(2) = \{c_2\}$.

    \begin{center}
        \begin{tikzpicture}[
            roundnode/.style={circle, draw=blue!80, fill=blue!70, very thick, minimum size=6mm},
            squarednode/.style={rectangle, draw=blue!50, fill=blue!40, very thick, minimum size=7mm},
            empty/.style={rectangle, draw=white, fill=white, minimum size = 5mm},
            ]
                

            \node[roundnode] (c1) {$c_1$};
            \node[squarednode] (h1) [below=of c1] {$h$};
            \node[empty] (t1) [above=0.5 of c1] {$t = 1$};


            \node[roundnode] (c2) [right=of c1, draw=red!80, fill=red!70] {$c_2$};
            \node[empty] (t1) [above=0.5 of c2] {$t = 2$};
        \end{tikzpicture}
    \end{center}

    Players' preferences are:
    \begin{align*}
        &\Hat{V}_h(c_2) = 2 > \Hat{V}_h(c_1) = 1 \\
        &U_{c_1}(h) = 3/2 \\
        &U_{c_2}(h) = 2
    \end{align*}
    with $w_c = 1$ and $w_h = 1/2$. Suppose that $q$ is a non-wasteful mechanism. Consider the action profile $a_h = (1,1)$ and strategy profile $\sigma'_h(c_1) = 0,\sigma'_h(c_2) = 1$. $q$ must specify $\mu_1(h) = \varnothing$ by individual rationality. Thus, $C(2|m_1) = \{c_1, c_2\}$ and $H(2|m_1) = \{h\}$. By non-wastefulness, $\mu_2(h) \in \{c_1, c_2\}$. By weak fairness, $\mu_2(h) = c_2$. $h$ exits with payoff $2 - 1/2 = 3/2$. Now, instead, suppose that $h$ reports truthfully under the same action profile, i.e., some report $\sigma_h(c_1) = \sigma_h(c_2) = 1$. Non-wastefulness implies that $\mu_1(h) = c_1$ and $h$ exits with payoff $1$. Thus, we have that for $a_h, \sigma_h, \sigma'_h$, $\Hat{V}^r_h(a_h,\sigma'_h) = 3/2 > \Hat{V}^r_h(a_h,\sigma_h) = 1$.
\end{proof}

%% file: Paper/Empirics.tex
\section*{IV \hspace{10pt} Simulation Results}

In this section, we simulate adoption markets from the U.S. empirical distribution observed in a national panel of foster care placements. We obtain children's preferences through an AI model that predicts an objective welfare measure for each match from another work, and we posit multiple preference models for homes. Finally, we benchmark each mechanism's placements, waiting costs, and justified envy. 

\input{Components/Data}
\input{Components/SimulatedPreferences}
\input{Components/MainResults}
\input{Components/AppliedInsights}

%% file: Components/Data.tex
\subsection*{IV.A \hspace{10pt} Data}

The Adoption and Foster Care Analysis and Reporting System (AFCARS, \qcitenp{afcars-data}) is a U.S. database tracking case-level information for all foster care children managed through entities required to report through Title IV-E. This includes nearly every child that has been placed in a foster home. This data spans all children from September 30th, 1999 to September 30th, 2021. \qcite{highsmith2024matching} cleans and transforms this data set into Placement Files that track data for 2,335,517 matches observed at the end of the AFCARS reporting period. We refer the reader to that work for the technical details. In addition, we restrict the Placement Files to only foster children with a stated case goal of adoption and terminated parental rights. We also include only non-relative foster homes which are the relevant adoptive homes in this setting.

We use this data to simulate two years ($T = 24$) of children and homes arriving to a child welfare system. Incentives for patience mechanically decrease as $t$ approaches $T$. In practice, dynamic incentives persist because most people anticipate an infinite or extremely long horizon. Our results focus on the first year where dynamic incentives are sufficiently strong. 

The set of simulated markets is $M$, and the number of simulations is $|M| = 100$. Each market $m$ is a randomly selected county with at least $120$ placements from 1999 to 2021 in the Placement Files. The set of children and homes arriving at $t$ in market $m$ are $c_m(t)$ and $h_m(t)$, respectively, which are randomly sampled without replacement from children and homes in the random county. We fix the number of arrivals to five per period so that each simulation has $120$ children and homes over the entire horizon. We focus on two characteristics for foster children: age and clinically diagnosed disability status. We observe both in the Placement Files and denote them $a(c)$ and $d(c)$, respectively.

%% file: Components/SimulatedPreferences.tex
\subsection*{IV.B \hspace{10pt} Preferences}

The main goal for most child welfare systems is to guarantee a safe, healthy permanent placement for all children. One threat to permanency is placement instability. After the matchmaker places a child in a home, the home becomes the child's \textit{pre-adoptive} home. The home must foster the child for six months before the adoption is finalized. If the home decides not to adopt the child, it returns the child to the welfare system, i.e., the placement is not stable. We refer to this as persistence to avoid confusion with the theoretic property of stability. Persistence is widely accepted as an objective metric of child welfare and is sometimes the only easily accessible and assessable metric (\qcitenp{rubin-2007}).

In our model, a safe and healthy placement is one where the child is acceptable to the home and the home is acceptable to the child. The matchmaker can define the acceptability threshold to include only homes that would adequately care for the child. Conditional on a placement being acceptable, the matchmaker's utility for a child $c$ matching with a home $h$, $U_c(h)$, is the predicted persistence probability. \qcite{highsmith2024matching} trains a Random Forest (RF) model to perform these predictions that achieves 77\% accuracy. Moreover, the same work demonstrates that match-specific characteristics strongly impact the predictions. We adopt this model and set
\begin{align*}
    U_c(h) = p(c,h)
\end{align*}
where $p(c,h)$ is the fraction of trees that vote in favor of persistence for the $(c,h)$ match from the RF model.

Next, we turn to home preferences. We do not have access to homes' preferences and are not aware of any current attempts to elicit cardinal utilities from homes. We assume that homes have utilities that depend on child age and disability, and we discuss below our method to assess the impact of imprecisely predicting home preferences. The utility of a home $h$ for a match $c$ is:
\begin{align*}
    \Hat{V}_h(c) = 1 - \eta\Big(\frac{a(c)}{18}\Big)^2 - \gamma \delta_h d(c)
\end{align*}
Across the United States, younger children almost always find placements much sooner than older children, and a county we interviewed suggested that, ceteris paribus, homes strongly prefer to adopt younger children. We model this as vertical differentiation among children: younger children are more desirable to homes. The cost for a child's age is quadratic as placement outcomes for children tend to worsen rapidly as a child approaches eighteen. Further, in our interviews, the same county reported that some homes are greatly dissuaded from adopting children with any disability, while others have more openness. We use this to model horizontal differentiation where $\delta_h \sim U[0,1]$. We test the following preference models:

\textit{Complete Information}---We set $\Hat{V}_h(c)$ to be the matchmaker's preference estimator as discussed below.

\textit{Vertical Differentiation}---($\eta = 1, \gamma = 0$). All homes have the same preferences across children that depends on the child's age.

\textit{Horizontal Differentiation}---($\eta = 0, \gamma = 1$). Homes have different preferences across children that depend on the child's disability status.

\textit{Mixed Differentiation}---($\eta = 1, \gamma = 1$). Homes have vertical differentiation that is the same across all children, and homes have horizontal differentiation across children.

We calibrate $w_c = \frac{14000}{12}$ as the median Title-IV payment per out-of-home foster youth per year to a county is approximately \$14,000. Finally, we set $w_h = 0.04$ so that a home's tradeoff for achieving the "ideal" match $\Hat{V}_h(c) = 1$ is about two years (25 months).

\subsubsection*{Estimating Home Preferences}

We cannot elicit nor structurally estimate home preferences. Instead, we attempt to predict them and aim to understand how effective our mechanisms can be when the matchmaker uses a preference estimator that is, at best, correlated with a home's true utility. When patience-free mechanisms operate with an inaccurate estimator, this invalidates the assumption---complete information---that allows the mechanism to eliminate dynamic incentives. Furthermore, when patience-freeness fails, homes may deviate from accepting placements which implies the failure of weak non-wastefulness. These two issues could compound to seriously harm the number of placements.

The Placement Files record the number of different placements a child has experienced. Movements between foster homes are usually not benign. \qcite{koh2014explains} documents that a majority of movements are initiated when a foster home requests a placement change because of child behavior problems. These behavior problems or underlying pathologies that cause them are likely to persist across homes, meaning that it is plausible that the number of placements for a child could correlate with match utility for homes.

We estimate a simple and frankly naive linear model:
\begin{align*}
    l(c) = \alpha + \beta_1a(c) + \beta_2d(c) + \epsilon_c
\end{align*}
$l(c)$ is the observed number of placements for child $c$. For a child $c$ in a market $m$, we transform this into home utility with the following function:
\begin{align*}
    V_h(c) = 1-  \frac{l(c)}{l_m}
\end{align*}
where $l_m$ is the 90th percentile of the number of placements among all children in the county of market $m$. The problems with this model are numerous: it is ad-hoc, causally unidentified, hides heterogeneity among home preferences, and more. We reiterate that our goal is not to construct a good estimator. Our objective is to demonstrate that a good estimator, while beneficial, is not necessary, and it is possible to implement our mechanisms with heuristic solutions rather than multi-year, intensive research efforts to identify preferences. We also foresee using far more characteristics of children and homes for more sophisticated machine learning models in real-world application.

We report the empirical root mean squared error (RMSE) for the matchmaker's estimator where:
\begin{align*}
    \text{RMSE} = \sqrt{\frac{\sum_{h \in H(t) \forall t} \sum_{c \in C(t) \forall t} \left( \Hat{V}_h(c) - V_h(c) \right)^2}{\sum_{t=1}^T |C(t)| * \sum_{t=1}^t |H(t)|}}
\end{align*}
and we will show the strength of the correlation between the matchmaker's estimator and home preferences.

Finally, how does a home decide to accept or decline a placement? If we compute the home's equilibrium strategy, the computation requires fixing an arbitrary action profile $a(T)$ then checking which homes maximize their utility. A home then updates its strategy if the change would yield higher immediate utility or higher utility in the future given the market's evolution and future placements. Yet, this changes the placements that might occur for other homes in the future. We would have to repeat this computation until $a(T)$ no longer updates. Unfortunately, this calculation is too expensive. 

Instead, we borrow an insight from patience-freeness. In the complete information case, accepting the first placement is a weakly dominant strategy for all homes. We assume that a home $h$ takes it as given that all other homes will accept their placements, and a home $h$ accepts its placement if, when all other homes accept their placements, it maximizes its utility through accepting the placement rather than waiting to accept a future placement.

\subsubsection*{Decentralized Matching Mechanism}

Most counties in the U.S. follow a family-driven matching model. When an adoptive child needs a placement, the child's caseworker announces the child to all licensed adoptive homes. The first home that responds to the request is prioritized for the match. In larger time scales such as a month, caseworkers might announce needs for multiple adoptions, and homes respond in a dictatorial order determined by a first-come-first-serve order. We model the time component as random so that this process is a Random Serial Dictatorship and use this mechanism to simulate decentralized matching outcomes. \\~\

\textbf{Step 1:} For any $t \geq 1$, using $\Hat{V}_h(\cdot)$ as preferences
\begin{enumerate}
    \item Randomly choose a home $h \in H(t|m_{t-1})$ that has not been chosen yet.
    \item Set $\mu_t(h)$ to be $h$'s favorite, acceptable $c \in C(t|m_{t-1})$ that has not been matched yet, if any exists.
    \item Repeat 1-2 until no children or no homes remain.\\~\
\end{enumerate} 
We refer to RSD as decentralization henceforth. One important advantage that decentralization confers is that it does not depend on the matchmaker's estimator. Yet, it is an unanswered empirical question if this translates into more adoptions than mechanisms which require preference data. Our main results provide an answer: decentralization is suboptimal.

%% file: Components/MainResults.tex
\subsection*{IV.C \hspace{10pt} Main Results}

\begin{table}[t!]
\centering
\caption{Statistical Measures for Estimator Fit}
\begin{tabular}{ l@{\hskip 5ex} c@{\hskip 5ex} c@{\hskip 5ex} c }
\hline
\hline
\textit{Measure} & Vertical & Horizontal & Mixed \\ \hline
$\rho$ & 0.933 & 0.344 & 0.839 \rule{0pt}{3ex} \\
$RMSE$ & 0.348 & 0.446 & 0.472 \rule{0pt}{3ex} \\
$R^2$ & -0.331 & -1.388 & -0.074 \rule{0pt}{3ex} \\ \hline \hline
\end{tabular}
\label{tab:estimator-fit}
\end{table}

We compare each centralized mechanisms against decentralization. Our outcomes focus on three measures: cumulative placements, cumulative child waiting costs, and the percentage of placed homes with justified envy at the time of placement. Fixing some $m_t$, the number of placed home with justified envy on $\mu_k \in m_t$ is $|J(k)|$. At some time $k \leq t$, the last measure is $\sum_{i=1}^{k} |J(i)| / H(i|m_{i-1})$.

For each simulation, we run each mechanism on the same market. We set HEDA's endowments at a width of $1/3$ for the first three months, then the width reduces to $w_h$ for all future months. We run simulations for every market $m \in M$ and average over all simulations in our graphs. Our simulations create four three-way panels where an outcome $Y_{m,t,q}$ depends on the market $m$, time period $t$, and the mechanism $q$, and each panel is for a different preference type.

Averaging across our one hundred markets in Table \ref{tab:estimator-fit}, the matchmaker's estimator has consistently high RMSE and consistently negative $R^2$, meaning that the sample mean achieves a closer fit than the estimator. However, our heuristic estimator approximates preferences that decrease as child age increases. The correlation coefficient $\rho$ demonstrates a strong fit for vertical and mixed preferences. Horizontal preferences do not correlate strongly with the estimator.

\subsubsection*{Outcomes under Weakly Fair and Patience-Free Mechanisms}

Our first main result contrasts HPDA and ADA to decentralized matching. We estimate the following OLS model to estimate the effects on placements. We run the model on a panel combining observations from all preference types with a column indicating the preference model. We drop observations other than HPDA, ADA, and decentralized matching for this comparison.
\begin{align}
    Y_{m,t,q} = \alpha &+ \beta_{0, e,g} PM_{e,g} + \beta_1 HPDA + \beta_2 ADA \\
    &+ \beta_{1,e,g} HPDA * PM_{e,g} + \beta_{2,e,g} ADA*PM_{e,g}\notag \\
    &+ \epsilon_{m,t,Q}\notag
\end{align}
where $PM_{e,g} =\mathbf{1}\{ \eta = e, \gamma = g \}$, and $(e,g) \in \{ (1, 0), (0, 1), (1,1) \}$. The outcome $Y_{m,t,q}$ is the number of placements in market $m$ at time $t$ under mechanism $q$. 

Our theoretical results prove that HPDA eliminates dynamic incentives when preferences are known, but they are silent on the empirical impact on the number of placements as well as the potential effects when cardinal utilities are not exactly known. In comparison to decentralization---where all homes know their own preferences and match accordingly---the latter force could create mismatches that harms the number of accepted placements when using centralized mechanisms. 

\begin{table}[t!]
    \centering
    \caption{Effect on Placements of HPDA and ADA versus Decentralization}
    \begin{tabular}{ l@{\hskip 5ex} c@{\hskip 5ex} l@{\hskip 5ex} c@{\hskip 5ex} l@{\hskip 5ex} c }
    \hline
    \hline
    \textit{Parameter} & & \textit{HPDA} & & \textit{ADA} \\ \hline
    $\alpha$ & 3.005\rule{0pt}{3ex} & $\beta_1$ & 1.370 & $\beta_2$ & 1.370 \\
    & (0.039) & & (0.056) & & (0.056) \\
    $\beta_{0,1,0}$ (Vertical) & 0.028 & $\beta_{1,1,0}$ & -0.212 & $\beta_{2,1,0}$ & -0.145\rule{0pt}{3ex} \\
    & (0.056) & & (0.079) & & (0.079) \\
    $\beta_{1,0,1}$ (Horizontal) & 0.133 & $\beta_{1,0,1}$ & -0.422 & $\beta_{2,0,1}$ & -0.576\rule{0pt}{3ex} \\
    & (0.056) & & (0.079) & & (0.079) \\
    $\beta_{1,1,1}$ (Mixed) & -0.044 & $\beta_{1,1,1}$ & -0.329 & $\beta_{2,1,1}$ & -0.354\rule{0pt}{3ex} \\
    & (0.056) & & (0.079) & & (0.079) \\ \hline \hline
    \end{tabular}
    \vspace{1ex}
    \label{tab:placements-hpda-rsd}
\end{table}

Table \ref{tab:placements-hpda-rsd} presents the results. HPDA and ADA increase average per-period placements by 1.37 or 45.67\% relative to the decentralized placements $\alpha = 3.005$ under known preferences\footnote{The set of unmatched agents is the same in every fair matching by the Rural Hospitals Theorem. Therefore, when all agents accept the first placement, the number of placements will be equivalent over time as is with HPDA and ADA under known preferences.}. This effect is highly statistically and economically significant. It is statistically impossible for the matchmaker's estimator to infer horizontal differentiation, yet even under this model, yet even in this worse case ($\beta_2 + \beta_{2,0,1}$), placements still increase by 0.794 (26.5\%). Running the same regression model from eq. (1) with $Y_{m,t,Q}$ equal to per-period waiting costs, we find that this increase in placements is equivalent to at least \$5,805 average decrease in monthly expenditure. 

\subsubsection*{Effects of Centralization}

The evidence above suggests that HPDA and ADA are improvements over decentralization. The evidence does not distinguish between the impact of centralization and the impact of dynamic incentives. A fair centralized mechanism could increase placements through better match targeting when some children and homes have aligned preferences\footnote{Consider a simple 2x2 case where $h_1$ and $c_1$ mutually prefer each other most, but $h_2$ prefers $c_1$ over $c_2$. Suppose that $h_2$ is unacceptable to $c_2$. A fair matching will always yield weakly more placements than a one-sided Pareto efficient matching.}. Moreover, a strictly non-wasteful mechanism will offer more placements than a patience-free mechanism. We turn to a comparison between SDA and RSD to examine the effects of centralization and fairness which implies strict non-wastefulness We begin with the same equation as (1) but modified for DA
\begin{equation}
    Y_{m,t,Q} = \alpha + \beta_{0, e,g} PM_{e,g} + \beta_1 SDA + \beta_{1,e,g} SDA * PM_{e,g} + \epsilon_{m,t,Q}
\end{equation}
Table \ref{tab:placements-da-rsd} shows results for all placements as the outcome. All coefficients for DA effects are either negative and statistically significant or zero and statistically indistinguishable from zero, except for $\beta_{1,0,1}$ which indicates that SDA has a small, positive impact under horizontal differentiation. These results indicate that fairness alone does not guarantee more nor faster child adoptions. Dynamic incentives emerge as the primary market inefficiency in this setting.
\begin{table}[t!]
    \centering
    \caption{Effect on Placements of SDA versus Decentralization}
    \begin{tabular}{ l@{\hskip 5ex} c@{\hskip 10ex} l@{\hskip 5ex} c }
    \hline
    \hline
    \textit{Parameter} &  \\ \hline
    $\alpha$ & 3.005 & $\beta_1$ & 0.037 \rule{0pt}{3ex} \\
    & (0.051) & & (0.072) \\
    $\beta_{0,1,0}$ (Vertical) & 0.028 & $\beta_{1,1,0}$ & -0.032\rule{0pt}{3ex} \\
    & (0.072) & & (0.103) \\
    $\beta_{0,0,1}$ (Horizontal) & 0.133 & $\beta_{1,0,1}$ & 0.348\rule{0pt}{3ex} \\
    & (0.072) & & (0.103) \\
    $\beta_{0,1,1}$ (Mixed) & -0.044 & $\beta_{1,1,1}$ & -0.243\rule{0pt}{3ex} \\
    & (0.072) & & (0.103) \\ \hline \hline
    \end{tabular}
    \label{tab:placements-da-rsd}
\end{table}
Our last consideration for SDA is its effect on persistence. One might expect that SDA would yield higher persistence rates than all other mechanisms. First, fairness ensures that there are no placed children that could achieve better predicted persistence at a home that would prefer a child over its match. Second, strict non-wastefulness will always propose matches to available homes even with prediction error, meaning that the market is thicker under SDA than a patience-free mechanism. We estimate eq. (1) and eq. (2) with average predicted persistence as the outcome. Surprisingly, Table \ref{tab:persistence} refutes these intuitions. In fact, HPDA has a statistically significant, positive effect on persistence (5\% increase) while SDA has a statistically significant, negative effect on persistence (3.6\% decrease). However, SDA's decrease is not consistent over preference models.

\begin{table}[hbt!]
\centering
\caption{Persistence for All Mechanisms}
\begin{tabular}{l@{\hskip 6ex}c@{\hskip 6ex}c@{\hskip 6ex}c@{\hskip 6ex}c}
\hline
\hline
\textit{Preferences} & Complete & Vertical & Horizontal & Mixed \\ \hline
\textit{Average persistence} \\
RSD & 0.89 & 0.9 & 0.86 & 0.9 \\
SDA & 0.86 & 0.9 & 0.94 & 0.88 \\
HPDA & 0.94 & 0.94 & 0.94 & 0.94 \\
ADA & 0.94 & 0.94 & 0.93 & 0.94 \\
HEDA & 0.92 & 0.92 & 0.91 & 0.91 \\ \hline
\end{tabular}

\label{tab:persistence}
\end{table}

This result implies that the placements that HPDA and ADA encourage homes to accept are the exact placements that are good for child welfare and that are likely to persist, indicating some amount of home welfare as well. Patience-freeness delivers more and better placements.

\subsubsection*{Comparing All Mechanisms}

\textit{Placements}---Our placement results are very promising for HEDA. Despite its lack of weak non-wastefulness, it differs very little from HPDA and ADA across all specifications. Similarly, the mechanisms---SDA and decentralization---that are not patience-free track each other. This suggests that patience-freeness is an important property for guaranteeing not only \textit{expedient} placements but also \textit{more} placements. Without patience-freeness, homes effectively always have an incentive to delay before encountering a satisficing placement. As $T$ approaches infinity, the constant expectation that a better placement might arrive tomorrow induces a consistent downward pressure on the amount of potential placements. Patience-freeness eliminates such incentives.


Furthermore, in Table \ref{tab:placements},  we show that these benefits confer directly to child populations that currently face difficulty in finding adoptive homes: teens and children with clinically diagnosed disabilities. We show average placement outcomes for all children, teenagers, and children with disabilities. The percent of teenagers placed at $t$ is the number of teenagers with an accepted placement divided by the number of teenagers in the market at $t$. The average percent of teenagers placed is the average over all $t$. The same holds for children with disabilities.

One stark result is under most preferences, neither SDA nor decentralization successfully place more than 15\% of teenagers nor children with disabilities into permanent homes. Most percentages are substantially lower, especially for teens when homes have vertical preferences and for children with disabilities when homes have horizontal preferences.

\begin{figure}[t!]
    \centering
    \includegraphics[width=1\linewidth]{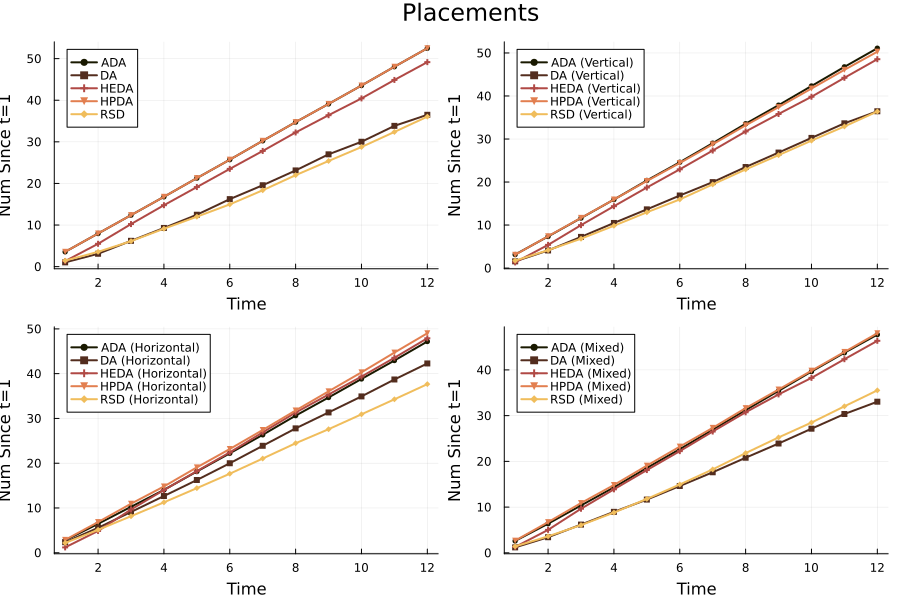}
    \caption{Placements, All Mechanisms}
    \label{fig:placements}
\end{figure}


Patience-free mechanisms improve outcomes for teenagers---HPDA and ADA so by 21-23 percentage points under vertical preferences and 15-16 percentage points under mixed preferences. Additionally, both mechanisms increase teen placements even with horizontal preferences because older age weakly correlates with clinically diagnosed disabilities in our Placement Files. SDA is slightly more successful at placing children with disabilities under horizontal preferences at 8\%. However, ADA improves by 14 pp and whereas HPDA and HEDA improve by 20 pp. Our mechanisms suffer from prediction error across all preference models with unknown preferences. All placements decline and waste increases.

\begin{table}[hbt!]
\centering
\caption{Placements for All Mechanisms}
\begin{tabular}{l@{\hskip 6ex}c@{\hskip 6ex}c@{\hskip 6ex}c@{\hskip 6ex}c}
\hline
\hline
\textit{Preferences} & Complete & Vertical & Horizontal & Mixed \\ \hline
\textit{Average placements} \\
RSD & 3.01 & 3.03 & 3.14 & 2.96 \\
SDA & 3.04 & 3.04 & 3.52 & 2.75 \\
HPDA & 4.38 & 4.19 & 4.09 & 4.0 \\
ADA & 4.38 & 4.26 & 3.93 & 3.98 \\
HEDA & 4.1 & 4.05 & 3.99 & 3.86 \\
\textit{Average percent, teens}\rule{0pt}{6ex} \\
RSD & 5\% & 2\% & 15\% & 2\% \\
SDA & 7\% & 2\% & 16\% & 2\% \\
HPDA & 33\% & 23\% & 22\% & 18\% \\
ADA & 33\% &25\% & 20\% & 17\% \\
HEDA & 20\% & 18\% & 18\% & 15\% \\
\textit{Average percent, disabled}\rule{0pt}{6ex} \\
RSD & 11\% & 14\% & 4\% & 7\% \\
SDA & 12\% & 15\% & 8\% & 6\% \\
HPDA & 46\% & 36\% & 24\% & 23\% \\
ADA & 46\% & 39\% & 18\% & 21\% \\
HEDA & 31\% & 30\% & 24\% & 23\% \\
\textit{Average waste}\rule{0pt}{6ex} \\
RSD & 4.51 & 5.24 & 5.59 & 5.86 \\
SDA & 4.51 & 4.51 & 4.51 & 4.51 \\
HPDA & 4.51 & 5.45 & 5.88 & 6.23 \\
ADA & 4.51 & 4.77 & 5.24 & 5.23 \\
HEDA & 7.05 & 7.29 & 7.28 & 7.82 \\ \hline
\end{tabular}
\vspace{0.1ex}

{\raggedright \small Note: waste is the number of homes in $H(t)$ that did not receive a placement offer. Some homes are unplaced in every stable match which explains DA's positive waste. \par}

\label{tab:placements}
\end{table}

Waste has its global maximum under HEDA with mixed preferences. This empirical observation is consistent with HEDA lacking weak non-wastefulness. We also see that HPDA's waste exceeds RSD when preferences are unknown, whereas ADA's waste never exceeds RSD, validating capacity efficiency. Nevertheless, this does not seem to significantly impact the number of placements that are accepted. This may be explained by the fact that ADA uses child-proposing DA, meaning that its matchings are home-pessimal, and it is more likely that a home can gain if they wait when preferences are unknown. Last, DA is strictly non-wasteful and offers as many placements as possible under all preference models. 

\textit{Waiting Costs}---Our results for child waiting costs (Figure \ref{fig:waitingcosts} in Appendix B) imposed on the child welfare system mirror the placement results as more placements mechanically implies less waiting. The gap tends to widen over time because the stock of unplaced children increases more quickly under SDA and decentralization. There is a natural limit to the increase in the gap as children eventually age out of the child welfare system which we do not capture in our model and figures.

\textit{Justified Envy}---SDA, HPDA, and ADA are weakly fair under complete information as our theoretical results proved. Even in this case, decentralized matching has a large amount of justified envy at over 60\% of homes. Unfortunately, this implies that some children are be able to receive placements that are more likely to persist, and there are homes that would adopt those children. HEDA exhibits about 10-20\% justified envy across all preferences. Correspondingly, Table \ref{tab:persistence} shows that its average persistence is slightly above RSD but below HPDA and ADA.

\begin{figure}[bt!]
    \centering
    \includegraphics[width=1\linewidth]{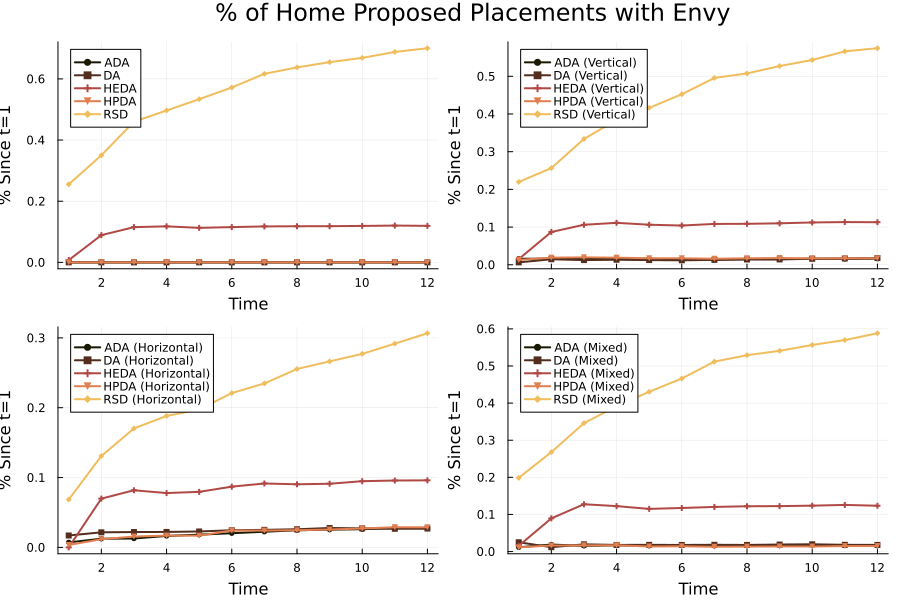}
    \caption{Justified Envy}
    \label{fig:envyfig}
\end{figure}

Under any preference model other than complete information, justified envy appears for every mechanism, although all centralized mechanisms have significantly less justified envy than decentralization. Justified envy is near zero for weakly fair mechanisms.

%% file: Components/AppliedInsights.tex
\subsection*{IV.D \hspace{10pt} Applied Insights}

We conclude with additional recommendations and concerns for mechanism implementation.

\subsubsection*{Centralization and Child Welfare}

Our simulations demonstrate that centralization is insufficient. Our initial hypothesis was that some inefficiencies in the child welfare system owed to a lack of coordination in adoption matching and that introducing a clearinghouse could ameliorate this. We find that a clearinghouse alone would not increase the number of placements nor would it improve child welfare. Our results suggest that dynamic incentives are the prevailing market inefficiency. This ought to inform practical efforts to decrease the number of waiting children: embracing any new matching mechanisms or tools that target match quality without addressing dynamic incentives is not likely to be substantially effective.

A centralized clearinghouse is costly. We are aware of two frictions. First, there are no information technology systems for adoption clearinghouses. The child welfare system is generally underfunded and does not direct many resources towards digital infrastructure. Implementing a centralized mechanism will require a well-designed system that must at minimum: provide platformed interaction and coordinated placements between caseworkers and homes, enforce placements, allow preference reporting, predict child persistence and home preferences, and integrate with back-end county systems. Our research demonstrates that limited preference elicitation and prediction are feasible; the remainder is up to the practitioner to implement. Last, the model and technology must be properly fitted. For example, although we propose a one-to-one matching mechanism, for some time scales it might be prudent to propose multiple matches to one home in one period. For HEDA, this does not theoretically differ because it is equivalent to offering only the best match from the set. Practically, it could be a benefit if there is prediction error or if per-match response rate is low. 

Second, as we discuss in section III.B, the regulatory laws governing who can make and to whom placement offers can be made substantially differ across counties. However, patience-free mechanisms only have participation guarantees when there are no outside options. It is important to consider how a clearinghouse might interact with the regulatory environment.

\subsubsection*{Prediction in Matching Theory}

Our next insight is that prediction is satisfactory. Much work in matching theory focuses on strategy-proofness which guarantees simple preference elicitation, but that is generally impractical for cardinal utilities and impossible when attempting to mitigate dynamic incentives. Economists have devised other empirical strategies to estimate preferences, but most if not all have onerous data requirements that a typical county is unlikely to satisfy. Comparatively, there are ready-made tools for prediction problems and a national dataset with millions of observations upon which tools can be trained. Our main results show that a predictor that is modestly correlated with preferences improves upon decentralized matching despite enormous degrees of error. We note that we did not analyze strategic incentives when there is prediction error, and some of the theoretical results might not hold in this case. We conjecture that predictor bias is irrelevant as it would not affect proposal orders under HEDA, but high variance in the errors might induce misreporting.

%% file: Paper/Conclusion.tex
\section*{V \hspace{10pt} Conclusion}

Child welfare systems lack organized attempts to centralize placement assignments for children. Previous efforts to construct match recommendation systems have not gone as hoped, yet the area still seems to hold significant promise to be an arena for market design, and many economists have studied adjacent problems to improve outcomes for children. We show that the 109 thousand children waiting for adoption in the United States, and certainly many more abroad, could have permanent homes---now. Nevertheless, even with centralization, matchmakers must choose carefully between mechanisms that prioritize different goals. Patience-free mechanisms always increase the number of placements and decrease waiting costs in our simulations. Strictly non-wasteful mechanisms emphasize proposing all possible placements, but they come at the price of lower persistence for adoptions, less placements for vulnerable children, and less placements overall. We offer at least one mechanism---HEDA---that increases adoptions, generates low to moderate waste, and satisfies strategy-proofness, meaning that it prevents incentives for homes that can care for high-needs or older children to report that they cannot. 

Several obstacles remain between current practice and smooth centralization. Caseworkers report that previous designs failed because of poor coding, match recommendations based on irrelevant characteristics, and many other operational concerns. Going forward, we desire that child welfare systems will conquer the challenge of designing better matching for the many waiting children everywhere.

%% file: Appendices/A.tex
\section*{Appendix A: Proofs}

Below, we prove the theorems and propositions that we did not prove in the main text.

\input{Components/Proofs/HPDA}

\input{Components/Proofs/ADA}

\textbf{Proposition \ref{pr:dominant}.} Accepting the first placement is a weakly dominant strategy for all homes under a patience-free mechanism.

\input{Components/Proofs/WeaklyDominant}

\textbf{Theorem \ref{th:HEDA}}. Home Endowed Deferred Acceptance is patience-free, individually rational, and strategy-proof.

\input{Components/Proofs/HEDA}

%% file: Components/Proofs/HPDA.tex
\textbf{Theorem \ref{th:HPDA}.} Home Penalized Deferred Acceptance is weakly fair, patience-free, and non-wasteful.
\begin{proof}
    First, we show that HPDA is weakly fair for any action profiles at any time. Then, we prove that HPDA is patience-free. Last, we show that HPDA is non-wasteful. 
    
    (i) HPDA is weakly fair for any $m_{t-1}$. Suppose for a contradiction that for some $m_{t-1}$, we have that $\mu_t \equiv \mu_{t,m}^q$ is not weakly fair. This implies that there exists some $c \in C(t|m_{t-1})$ and $h \in H(t|m_{t-1})$ with $\mu_t(h) \neq h$ such that $U_c(h) > U_c(\mu_t(c))$ and $V_h(c) > V_h(\mu_t(h))$. However, this implies that $h$ did not have its preferences truncated at $t$, otherwise $\mu_t(h) = h$. Therefore, during the operation of HPDA, $h$ proposed to $c$. Since $\mu_t(c) \neq h$, this implies that $c$ rejected $h$. This contradicts $U_c(h) > U_c(\mu_t(c))$. 

    (ii) A non-individually rational match is impossible because no home would propose to an unacceptable child, and no child would hold a proposal from an unacceptable home.

    (iii) HPDA is patience-free. Toward a contradiction, suppose that it is not. Then, for some $m_{t-1}$, there exists an $h \in H(t|m_{t-1})$ such that $V_h^t(\mu_t(h)) > V_h^k(\mu_k(h))$ for some $k < t$ and $\mu_k(h) \neq h$. However, $V_h^t(\mu_t(h)) > V_h^k(\mu_k(h))$ implies that $p_h^t = 1$, but then $h$ did not propose to $\mu_t(h)$, a contradiction.

    (iv) Last, HPDA is non-wasteful. Suppose, for a contradiction, the contrary. Since $a_h^t = 1$ for all $h,t$, if $\mu_k(h) \neq h$ for some $k < t$, then $h \notin H(t|m_{t-1})$. Therefore, it cannot be that any home has its preferences truncated. But by our supposition, there exists some $c \in C(t|m_{t-1})$ and $h \in H(t|m_{t-1})$ such that $U_c(h) \geq 0$, $V_h(c) \geq 0$, $\mu_t(c) = c$, and $\mu_t(h) = h$. The child and homes' non-negative utilities for the match and $\mu_t(h) = h$ imply that, at some point, $h$ proposed to $c$ and $c$ had no proposals. However, $\mu_t(c) = c$ implies that $c$ rejected $h$, which contradicts $U_c(h) \geq 0$. This completes the theorem.
\end{proof} 

%% file: Components/Proofs/ADA.tex
\textbf{Theorem \ref{th:ADA}.} Ascending Deferred Acceptance is weakly fair, patience-free, non-wasteful, and capacity efficient.
\begin{proof}
(i) ADA is weakly fair for any $m_{t-1}$. Suppose for a contradiction that for some $m_{t-1}$, we have that $\mu_t$ is not weakly fair. This implies that there exists some $c \in C(t|m_{t-1})$ and $h \in H(t|m_{t-1})$ with $\mu_t(h) \neq h$ such that $U_c(h) > U_c(\mu_t(c))$ and $V_h(c) > V_h(\mu_t(h))$. However, this implies that $h$ did not have its preferences truncated at $t$ in the last round $n$, otherwise $\mu_t(h) = h$. Therefore, during round $n$, $c$ proposed to $h$. Since $\mu_t(c) \neq h$, this implies that $h$ rejected $c$. This contradicts $V_h(c) > V_h(\mu_t(h))$.

    (ii) ADA is individually rational. The proof is exactly the same as in HPDA.

    (iii) ADA is patience-free. Toward a contradiction, suppose that it is not. Then, for some $m_{t-1}$, there exists an $h \in H(t|m_{t-1})$ such that $V_h^t(\mu_t(h)) > V_h^k(\mu_k(h))$ for some $k < t$ and $\mu_k(h) \neq h$. However, $V_h^t(\mu_t(h)) > V_h^k(\mu_k(h))$ implies that $h$'s preferences are truncated in some round $n$, but then $h$ did not propose to $\mu_t(h)$ in round $n+1$, a contradiction.

    (ii) ADA is non-wasteful. The proof is exactly the same as in HPDA.

    (v) ADA is capacity efficient. First, we note key results. 

    \qcite{gusfield_irving_1989}, Theorem 1.2.3. \textit{In the man-optimal stable matching, each woman has the worst partner she can have in any stable matching.} Let $H^n = \{ h \in H(t|m_{t-1}) : \exists c' \in C(t|m_{t-1}) \text{ such that } \tilde{V}_h^n(c') > 0 \ \}$. The matching $\mu_t^n$ is the outcome of child-proposing DA in round $n$, which is the child-optimal fair (stable) matching on $C(t|m_{t-1})$ and $H^n$. Therefore, each home in $H^n$ has the worst child match it can have in any fair matching on $C(t|m_{t-1})$ and $H^n$.

    \qcite{crawford1991comparative}, Theorem 1. (We translate this theorem directly into our setting.) \textit{Let $H' \subseteq H$, $\mu$ the child-optimal stable matching on $C$ and $H$, and $\mu'$ the child-optimal stable matching on $C$ and $H'$. Then $V_h(\mu(h)) \leq V_h(\mu'(h))$ for all $h \in H'$.} Enlarging the set of homes makes all homes weakly worse off; conversely, shrinking the set of homes makes all homes weakly better off.

    Suppose, for a contradiction, that $\mu_t$ is not capacity efficient for some $m_{t-1}$. Then, there exists some $\mu' \in \phi_{t,m}$ such that $\kappa(\mu'|H(t|m_{t-1})) \subset \kappa(\mu_t|H(t|m_{t-1}))$. Let the complement of $\kappa$ be
    \[\kappa'(\mu|H) = \{ h \in H : \mu(h) \neq h\}\]
    and
    \[D = \kappa'(\mu'|H(t|m_{t-1})) - \kappa'(\mu_t|H(t|m_{t-1}))\]
    Since $\kappa'$ is the complement of $\kappa$, $D$ is non-empty. Consider the earliest round $n$ such that $\tilde{V}^{n+1}_h(c') < 0$ for all $c' \in C(t|m_{t-1})$ for some $h \in D$. It is immediate that $\kappa'(\mu_t|H(t|m_{t-1})) \cup D \subseteq H^n$. Therefore, $\kappa'(\mu'|H(t|m_{t-1})) \subseteq H^n$. Note again that $\mu_t^n$ is the child optimal (home pessimal) fair matching on $C(t|m_{t-1})$ and $H^n$, and define $\mu''$ as the child optimal (home pessimal) fair matching on $C(t|m_{t-1})$ and $\kappa'(\mu'|H(t|m_{t-1}))$. So, for all $h \in D$, we have
    \[V_h(\mu''(h)) \geq V_h(\mu_t^n(h)) > V_h(\mu_k(h))\]
    for some $\mu_k \in m_{t-1}$. The first inequality follows by \qcite{crawford1991comparative}, and the second by assumption that $h$ is truncated in round $n$.
    
    $\mu'$ is weakly fair. So, it is fair on $C(t|m_{t-1})$ and $\kappa'(\mu'|H(t|m_{t-1}))$ (if not, the contradiction is immediate since there will be a blocking pair $c,h$ where $\mu'(h) \neq h$). Since $\mu''$ is home-pessimal, by \qcite{gusfield_irving_1989}, we have
    \[V_h(\mu'(h')) \geq V_h(\mu''(h'))\]
    for all $h' \in \kappa'(\mu'|H(t|m_{t-1}))$. Combining this and the above inequalities:
    \[V_h(\mu'(h)) > V_h(\mu_k(h))\]
    This contradicts $\mu' \in \phi_{t,m}$.
    
    This concludes the Theorem.

\end{proof}

%% file: Components/Proofs/WeaklyDominant.tex
\begin{proof}
    A mechanism $q$ and actions $a(T)$ define a sequence of market matchings where $m_{1,a(1)}^q = (a(1), q(M^1))$, and 
    \[m^q_{t, a(t)} = (a(t),  m^q_{t-1,a(t-1)}, q(m^q_{t-1,a(t-1)}, M^t))\]
    A matching at period $t$ for these market matchings is $\mu_{t,a(T)}^q = q(m^q_{t-1,a(t-1)}, M^t)$.
    Fix a patience-free mechanism $q$ and an action profile $a(T)$ where $a_h^t = 1$ for all $t$. Let $t_h$ be the first period when $\mu_{t,a(T)}^q(h) \neq h$. Since $\mu_t = q(m_{t-1,a(t-1)}^q, M^t)$, we have that the matching at time $t$ does not depend on actions at any time $t' \geq t$. Therefore, for any alternative action $b^{t'}_h$ where $t' \geq t_h$, we have that $\mu_{t_h}$ does not change. Define $b(T)$ equal to $a(T)$ except for $b^{t'}_h$. By patience-freeness, $V_h^{t'}(\mu_{t',b(T)}^q(h)) \leq V_h^{t_h}(\mu_{t_h,b(T)}^q(h)) = V_h^{t_h}(\mu_{t_h,a(T)}^q(h))$. This proves the Proposition.
\end{proof}

%% file: Components/Proofs/HEDA.tex
\begin{proof}
    First, we show that HEDA is individually rational for any action profiles at any time. Then, we prove that HEDA is patience-free. Last, we show that HEDA is strategy-proof. For individual rationality and patience-freeness, we take as given the input preferences and prove the properties with respect to the input.
    
    (i) An non-individually rational match is impossible because no home would propose to an unacceptable child, and no child would hold a proposal from an unacceptable home.

    (ii) HEDA is patience-free. Toward a contradiction, suppose it is not. Then, for some $m_{t-1}$,, there exists an $h \in H(t|m_{t-1})$ such that $V_h^t(\mu_t(h)) > V_h^k(\mu_k(h))$ for some $k < t$ and $\mu_k(h) \neq h$. Since $\mu_t(h) \neq h$, we have that $e_h^t(\mu_t(h)) = 1 \iff V_h^t(\mu_t(h)) \in B_t(h)$. Similarly, $V_h^k(\mu_k(h)) \in B_k(h)$. This is a contradiction because $k < t$ implies $V_h^t(\mu_t(h)) < V_h^k(\mu_k(h))$.
    

    (iii) HEDA is strategy-proof. Take any arbitrary $a_{-h}(T),\sigma_{-h}$. Let $\sigma_h$ be truthful and $a_h(T)$ be compliant, i.e., $a_h^t = 1 \hspace{5pt} \forall t$. Denote $h$'s realized utility under this plan as $\Hat{V}_h^r(a(T), \sigma)$.
    
    Consider any alternative pair for $h$: $(\Hat{a}_h(T),\Hat{\sigma}_h)$. Then $\Hat{a}(T)$ is the profile with the specified action for $h$ and $\Hat{a}_{h'}(T) = a_{h'}(T) \hspace{5pt} \forall h' \neq h$. Likewise for $\Hat{\sigma}$. Denote $h$'s realized utility under the alternative pair as $\Hat{V}_h^r(\Hat{a}(T), \Hat{\sigma})$. We prove the theorem by showing that in the first period where matchings differ under the two strategy pairs, it must be that $h$'s match cannot improve. Under HEDA, $h$ maximizes its utility through accepting the earliest possible match, so any subsequent match under the alternative pair is also worse. The theorem follows. 
    
    We use a few last pieces of notation for the proof. The market matching at $k$ given the truthful pair is $m_k = (m_{k - 1}, a(k), q(m_{k-1}, M^k(\sigma))$ where $m_1 = (a(1), q(M^1(\sigma)))$. The matching at time $k$ under this alternative pair is $\mu_k \equiv q(m_{k-1}, M^k(\sigma))$.
    
    The market matching at $k$ given the alternative pair is $\hat{m}_k = (\hat{m}_{k - 1}, \hat{a}(k), q(\hat{m}_{k-1}, M^k(\hat{\sigma}))$ where $\hat{m}_1 = (\hat{a}(1), q(M^1(\hat{\sigma})))$. The matching at time $k$ under this alternative pair is $\hat{\mu}_k \equiv q(\hat{m}_{k-1}, M^k(\hat{\sigma}))$. Abusing notation, we redefine the sets of children and homes at each time as those endogenously determined by the mechanism, actions, and strategies:
    \begin{align*}
        C(t) &\equiv C(t|m_{t-1}) \\
        H(t) &\equiv H(t|m_{t-1}) \\
        A^C(t) &= \bigcup_{k=1}^t c(t) - C(t)\\
        A^H(t) &= \bigcup_{k=1}^t h(t) - H(t)
    \end{align*}
    and
    \begin{align*}
        \Hat{C}(t) &\equiv C(t|\hat{m}_{t-1}) \\
        \Hat{H}(t) &\equiv H(t|\hat{m}_{t-1}) \\
        \Hat{A}^C(t) &= \bigcup_{k=1}^t c(t) - \Hat{C}(t)\\
        \Hat{A}^H(t) &= \bigcup_{k=1}^t h(t) - \Hat{H}(t)
    \end{align*}
    Take the first period $k$ where $\mu_k \neq \Hat{\mu}_k$, if one exists. If one does not exist, the theorem follows trivially. Otherwise, we first prove that $A^C(k') = \Hat{A}^C(k')$ and $A^H(k') = \Hat{A}^H(k') \hspace{5pt} \forall k' < k$. Suppose, for a contradiction, that this is not true. If it is not true for the latter, then take the earliest $k'$ where there exists a $h' \in A^H(k'), h' \notin \Hat{A}^H(k')$ or $h' \notin A^H(k'), h' \in \Hat{A}^H(k')$. If $h' = h$, this would imply that $h \in A^H(k'), h \notin \Hat{A}^H(k')$ (since $h$ must accept the placement under the truthful, compliant pair), but then $h$'s realized utility is greater under the truthful, compliant pair by (ii). Thus, we can assume WLOG that $h' \neq h$. However, $\mu_{k'}(h') = \Hat{\mu}_{k'}(h')$ by assumption, and $\Hat{a}_{h'}^k = a_{h'}^k$. This is a contradiction, so, WLOG, $A^H(k') = \Hat{A}^H(k')$. Next, suppose this is not true for the former, i.e., take the earliest $k'$ where there exists a $c \in A^C(k'), c \notin \Hat{A}^C(k')$ or $c \notin A^C(k'), c \in \Hat{A}^C(k')$. If $c = \mu_{k'}(h)$, then, again, $h$'s realized utility is greater under the truthful, compliant pair so we can assume, WLOG, that $c \neq \mu_{k'}(h)$. Yet, $\mu_{k'}(c) = \Hat{\mu}_{k'}(c)$ by assumption, and $\Hat{a}_{\Hat{\mu}_{k'}(c)}^k = a_{\mu_{k'}(c)}^k$. This contradicts $A^C(k') \neq \Hat{A}^C(k')$, so, WLOG, $A^C(k') = \Hat{A}^C(k')$. These two facts imply that $C(k) = \Hat{C}(k)$ and $H(k) = \Hat{H}(k)$.
    
    We now have that the set of children and homes available to match at $k$ are the same under either strategy pair. Returning to the period $k$, we split into three cases. Either $\mu_k(h) = \Hat{\mu}_k(h) = h$, $\mu_k(h) = \Hat{\mu}_k(h) \neq h$ or $\mu_k(h) \neq \Hat{\mu}_k(h)$. 
    
    \textbf{Case 1}: Suppose, for a contradiction, that it is true. Then for some $h' \neq h$, $\mu_k(h') \neq \Hat{\mu}_k(h')$. If $\Tilde{V}_{h'}(\Hat{\mu}_k(h')) > \Tilde{V}_{h'}(\mu_k(h'))$, this implies the existence of some $h'' \neq h',h$ with $\mu_k(h'') \equiv c'' \neq h''$ and $\Hat{\mu}_k(h') = c''$. If not, then $\mu_k(c'') = c''$, implying that $h'$ proposed to this child, but the child rejected $h'$ under $a(T),\sigma$. Then, $h'$ is unacceptable to this child, contradicting $\Hat{\mu}_k(c'') = h'$. Furthermore, it must be that $\Tilde{V}_{h''}(\Hat{\mu}_k(h'')) > \Tilde{V}_{h''}(\mu_k(h''))$. If not, this implies $h''$ proposed to $c''$ under $\Hat{a}(T), \Hat{\sigma}$ but was rejected in favor of $h'$. But then $h'$ and $c''$ form a blocking pair under $a(T), \sigma$, meaning that $\mu_k$ is not stable w.r.t the constructed preferences, which is a contradiction.

    Define $S \equiv \{ \Hat{h} : \Tilde{V}_{\Hat{h}}(\Hat{\mu}_k(\Hat{h})) > \Tilde{V}_{\Hat{h}}(\mu_k(\Hat{h})) \}$ and suppose it is non-empty. Consider the first round of proposals at time $k$ under $a(T), \sigma$ where some $\Hat{h} \in S$ is rejected by $\Hat{c} \equiv \Hat{\mu}_k(\Hat{h})$ in favor of some $\Tilde{h}$. If $\Tilde{h}$ not has proposed to $\Hat{\mu}_k(\Tilde{h})$ or is unmatched, this implies $\Tilde{V}_{\Tilde{h}}(\Hat{c}) > \Tilde{V}_{\Tilde{h}}(\Hat{\mu}_k(\Tilde{h}))$ and $U_{\Hat{c}}(\Tilde{h}) > U_{\Hat{c}}(\Hat{h})$. Then, $\Hat{c}$ and $\Tilde{h}$ form a blocking pair on $\Hat{\mu}_k$, which contradicts DA's stability on the constructed preferences. Therefore, it must be that $\Tilde{h}$ has proposed to and been rejected by $\Hat{\mu}_k(\Tilde{h}) \implies \Tilde{V}_{\Tilde{h}}(\Hat{\mu}_k(\Tilde{h})) > \Tilde{V}_{\Tilde{h}}(\Hat{c}) >  \Tilde{V}_{\Tilde{h}}(\mu_k(\Tilde{h})) \implies \Tilde{h} \in S$. However, by assumption that this is the first round that any home in $S$ was rejected, it cannot be that $\Tilde{h}$ has been rejected yet. This is a contradiction; $S$ must be an empty set.

    Instead, if $\Tilde{V}_{h'}(\Hat{\mu}_k(h')) < \Tilde{V}_{h'}(\mu_k(h'))$, there must exist some $h'' \neq h',h$ with $\Hat{\mu}_k(h'') \equiv c'' \neq h''$ and $\mu_k(h') = c''$. We also we have that $\Tilde{V}_{h''}(\Hat{\mu}_k(h'')) < \Tilde{V}_{h''}(\mu_k(h''))$. All of the above follows from the same logic as before. We can define $S' \equiv \{ \Hat{h} : \Tilde{V}_{\Hat{h}}(\Hat{\mu}_k(\Hat{h})) < \Tilde{V}_{\Hat{h}}(\mu_k(\Hat{h})) \}$ and suppose that it is non-empty. Consider the first round of proposals at time $k$ under $\Hat{a}(T), \Hat{\sigma}$ where some $\Hat{h} \in S'$ is rejected by $\Hat{c} \equiv \mu_k(\Hat{h})$ in favor of some $\Tilde{h}$. If $\Tilde{h}$ not has proposed to $\mu_k(\Tilde{h})$ or is unmatched, this implies $\Tilde{V}_{\Tilde{h}}(\Hat{c}) > \Tilde{V}_{\Tilde{h}}(\mu_k(\Tilde{h}))$ and $U_{\Hat{c}}(\Tilde{h}) > U_{\Hat{c}}(\Hat{h})$. Then, $\Hat{c}$ and $\Tilde{h}$ form a blocking pair on $\mu_k$, which contradicts DA's stability on the constructed preferences. Therefore, it must be that $\Tilde{h}$ has proposed to and been rejected by $\mu_k(\Tilde{h}) \implies \Tilde{V}_{\Tilde{h}}(\mu_k(\Tilde{h})) > \Tilde{V}_{\Tilde{h}}(\Hat{c}) > \Tilde{V}_{\Tilde{h}}(\Hat{\mu}_k(\Tilde{h})) \implies \Tilde{h} \in S'$. However, by assumption that this is the first round that any home in $S'$ was rejected, it cannot be that $\Tilde{h}$ has been rejected yet. This is a contradiction; $S'$ must be an empty set.

    Since $S$ and $S'$ are empty, for every $h' \neq h$, it must be that $\Tilde{V}_{h'}(\mu_k(h')) = \Tilde{V}_{h'}(\Hat{\mu}_k(h'))$. Since preferences are strict, it cannot be that $\mu_k(h') \neq \Hat{\mu}_k(h')$ for any $h'$. Hence, case 1 is impossible.

    \textbf{Case 2:} In this case, since $h$ always accepts the first match under $a(T)$, and $\mu_k(h) = \Hat{\mu}_k(h) \neq h$, it must be that this is the first period that $h$ receives a match under either $a(T),\sigma$ or $\Hat{a}(T),\Hat{\sigma}$. Then, $\Hat{V}_h^r(a(T), \sigma) = \Hat{V}_h^k(\mu_k(h))$. If $h$ accepts the match under $\Hat{a}(T)$, $h$ receives equivalent utility. If not, $h$ receives some utility $\Hat{V}_h^r(\Hat{a}(T), \Hat{\sigma}) = \Hat{V}_h^n(\Hat{\mu}_n(h)) < \Hat{V}_h^r(a(T), \sigma) = \Hat{V}_h^k(\mu_k(h))$ in a period $n > k$. If $h$ is not matched at $n$, the inequality holds trivially. If $h$ is matched at $n$, we have that $\Hat{\mu}_n(h) \neq h \iff e_h^n(\Hat{\mu}_n(h)) = 1 \iff \Tilde{V}_h^n(\Hat{\mu}_n(h)) \in B_n(h) \implies \Hat{V}_h^n(\Hat{\mu}_n(h)) \in B_n(h)$. The last implication holds because $e_h^n(\Hat{\mu}_n(h)) = 1 \implies \Tilde{V}_h(\cdot) > 0 \implies \Tilde{V}_h(\cdot) = \Hat{V}_h(\cdot) \implies \Tilde{V}_h^n(\cdot) = \Hat{V}_h^n(\cdot)$. Similarly, $\Hat{V}_h^k(\mu_k(h)) \in B_k(h)$. By definition of HEDA, $h$ must have higher realized utility from accepting the match in period $k$, and the inequality holds.

    \textbf{Case 3:} Last, we adapt \qcite{roth-2017}'s proof to show that $h$'s utility cannot decrease.  To ease notation, we write $v_h(\cdot) \equiv \Tilde{V}_h(\cdot|\sigma_h)$ and $\Hat{v}_h(\cdot) \equiv \Tilde{V}_h(\cdot|\Hat{\sigma}_h)$ (note for all $h' \neq h$, constructed preferences are the same under $\sigma$ and $\Hat{\sigma}$). Define $S \equiv \{ \Hat{h} : \Hat{v}_{\Hat{h}}(\Hat{\mu}_k(\Hat{h})) > v_{\Hat{h}}(\mu_k(\Hat{h})) \}$ and $R \equiv \{ c : \Hat{\mu}_k(c) \in S \}$. 

    For a contradiction, suppose that $h \in S$. (A), we show that $c \in R \iff \mu_k(c) \in S$. (B) We show that a contradiction arises from our supposition that $h \in S$.

    Consider (A $\implies$). Suppose that $h' \in S \iff c' \equiv \Hat{\mu}_k(h') \in R$. Let $h'' \equiv \mu_k(c')$. If $h'' = h$, the statement follows trivially. If not, we know that $v_{h''}(\cdot) = \Hat{v}_{h''}(\cdot)$. Furthermore, $h' \in S \iff \Hat{v}_{h'}(c') = v_{h'}(c') > v_{h'}(\mu_k(h'))$ where the equality follows because if $h' = h$, $\Hat{v}_{h'}(c') > v_{h'}(\mu_k(h')) \geq 0 \implies \Hat{v}_{h'}(c') = \Hat{V}_{h'}(c') = v_{h'}(c')$, and it follows trivially if $h' \neq h$. Then, because DA is stable on the constructed preferences and $\mu_k(h') \neq c'$, it must be that $U_{c'}(h'') > U_{c'}(h')$. However, again by DA's stability on $\Hat{a}(T),\Hat{\sigma}$, it must be that $V_{h''}(\Hat{\mu}_k(h'')) > V_{h''}(c') \iff \mu_k(c') = h'' \in S$

    (A $\impliedby$) $\mu_k(c) \in S \implies c \in R$ if and only if the contrapositive is true, that is, $c \notin R \implies \mu_k(c) \notin S$. By the above proof, for every $h' \in S$, there exists exactly one $h'' \in S$ with $\mu_k(h'') \equiv c''$ such that $\Hat{\mu}_k(h') = c''$. Define the $h''$ satisfying this for $h'$ as $s(h')$. It must be that $\mu_k(s(h')) = c'' \in R$ because $\Hat{v}_{h'}(c'') > v_{h'}(\mu_k(h'))$. Furthermore, for two $h_1 \neq h_2$ it cannot be that $s(h_1) = s(h_2)$ because this would imply $\Hat{\mu}_k(h_1) = \Hat{\mu}_k(h_2)$. 
    
    Suppose, for a contradiction, that some $c \notin R$ and $h_c \equiv \mu_k(c) \in S$. It cannot be that $s(h_p) = h_c$ for any $h_p \in S$ as $s(h_p) = h_c \implies \mu_k(h_c) \in R$. Notice, then, $S$ contains $|S|$ elements that have some $s(\cdot)$. But then there are $|S| - 1$ homes that can satisfy $s(\cdot)$ at most. By the pigeonhole principle, at least one home must satisfy $s(\cdot)$ for at least two homes. As we proved, this is impossible. Hence, the contrapositive must be true, and (A $\impliedby$) is true.

    (B) Consider the last round that some arbitrary $h_l \in S$ proposes under $a(T), \sigma$. By definition, $h_l$ must propose to $c_l \equiv \mu_k(\Hat{h})$. For all $h' \in S$, $\Hat{\mu}_k(h')$ rejects $h'$. Furthermore, $h' \in S \implies \mu_k(h') \in R \implies \mu_k(h') \neq h'$, so it must be that every $h'$ has a proposal after its rejection by $\Hat{\mu}_k(h')$. $h_l \in S \implies c_l \in R$ which implies that $c_l$ receives a proposal from $\Hat{h} \equiv \Hat{\mu}_k(c_l)$. Since this is the last round of proposals for all homes in $S$, $c_l$ rejected $\Hat{h}$ in favor of some $h_s \notin S$ in a previous round of proposals (if $h_s \in S$, this means $h_s$ is rejected in favor of $h_l$ and proposes to another child in a future round, which contradicts this being the last round of proposals for homes in $S$). Then, $U_{c_l}(h_s) > U_{c_l}(\Hat{h})$. Moreover, $h_s$ proposed to $c_l$ before $\mu_k(h_s)$, so $v_{h_s}(c_l) > v_{h_s}(\mu_k(h_s)) > v_{h_s}(\Hat{\mu}_k(h_s))$. Since $h_s \notin S \implies h_s \neq h$, $v_{h_s}(\cdot) = \Hat{v}_{h_s}(\cdot)$, and $h_s$ must propose to $c_l$ under $a(T),\sigma$ as well as $\Hat{a}(T),\Hat{\sigma}$. This implies that $h_s$ and $c_l$ form a blocking pair on $\Hat{\mu}_k$, a contradiction. Therefore, it cannot be that $h \in S$.

    Finally, $h \notin S \implies \Hat{V}_h(\mu_k(h)) = v_h(\mu_k(h)) > \Hat{v}_h(\Hat{\mu}_k(h)) = \Hat{V}_h(\Hat{\mu}_k(h))$ where the last equality again holds because $\Hat{v}_h(\Hat{\mu}_k(h)) \geq 0$ by IR. Since $v_h(\mu_k(h)) > 0$, $h$ must receive a match. Since $a_h(T)$ is compliant and this is the first period where the matchings differ under $a(T), \sigma$ and $\Hat{a}(T), \Hat{\sigma}$, this implies this is $h$'s first match under either pair. $h$ will accept this match and receive utility $\Hat{V}_h^r(a(T), \sigma) = \Hat{V}_h^k(\mu_k(h)) > \Hat{V}_h^k(\Hat{\mu}_k(h))$. Hence, even if $h$ accepts the match under $\Hat{a}(T), \Hat{\sigma}$, $h$ receives strictly lower utility. If $h$ waits until a future period to accept a match under $\Hat{a}(T), \Hat{\sigma}$, $h$ receives lower utility by Case 2. This completes the theorem.

\end{proof} 

%% file: Appendices/B.tex
\section*{Appendix B: Additional Figures}

\begin{figure}[H]
    \centering
    \includegraphics[width=1\linewidth]{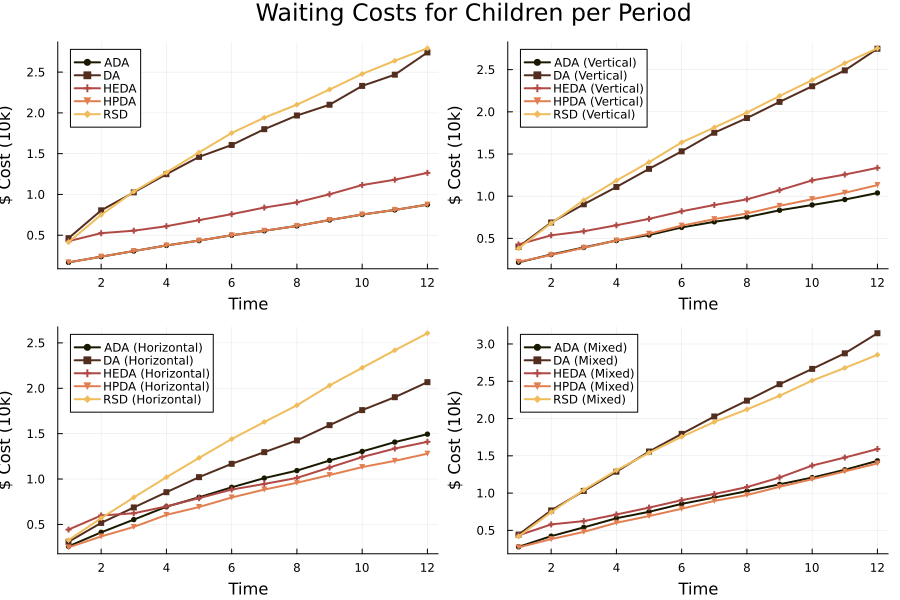}
    \caption{Waiting Costs}
    \label{fig:waitingcosts}
\end{figure}